\documentclass[10pt, journal, twocolumn, oneside, final, romanappendices]{IEEEtran}

\usepackage{color}
\usepackage{multirow}
\usepackage{graphicx}
\usepackage{epstopdf}
\usepackage[cmex10]{amsmath}
\usepackage{amssymb}

\usepackage{multirow}
\usepackage{amsthm}
\usepackage{cite}
\usepackage{flushend}
\usepackage{acronym} 
\usepackage{algorithm, algcompatible}

\newcommand{\diag}{\mathop{\rm diag}}
\newcommand{\Tr}{\mathop{\rm tr}}

\newtheorem{remark}{Remark}
\newtheorem{lemma}{Lemma}

\newtheorem*{theorem}{Theorem}
\newtheorem{corollary}{Corollary}
\newtheorem*{proposition}{Proposition}

\algnewcommand\INPUT{\item[\textbf{Input:}]}
\algnewcommand\OUTPUT{\item[\textbf{Output:}]}

\allowdisplaybreaks

\begin{document}

\title{Optimal Power Allocation and Active Interference Mitigation for Spatial Multiplexed MIMO Cognitive Systems}

\author{Nikolaos~I.~Miridakis, Minghua~Xia, \IEEEmembership{Member,~IEEE} and Theodoros~A.~Tsiftsis,~\IEEEmembership{Senior Member,~IEEE}
\thanks{N. I. Miridakis is with the Department of Computer Systems Engineering, Piraeus University of Applied Sciences, 12244, Aegaleo, Greece (e-mail: nikozm@unipi.gr).}
\thanks{M. Xia is with the School of Electronics and Information Technology, Sun Yat-sen University, Guangzhou, 510006, China (e-mail: xiamingh@mail.sysu.edu.cn).}
\thanks{T. A. Tsiftsis is with the School of Engineering, Nazarbayev University, Astana 010000, Kazakhstan (e-mail: theodoros.tsiftsis@nu.edu.kz).}
}


\maketitle

\begin{abstract}
In this paper, the performance of an underlay multiple-input multiple-output (MIMO) cognitive radio system is analytically studied. In particular, the secondary transmitter operates in a spatial multiplexing transmission mode, while a zero-forcing (ZF) detector is employed at the secondary receiver. Additionally, the secondary system is interfered by multiple randomly distributed single-antenna primary users (PUs). To enhance the performance of secondary transmission, optimal power allocation is performed at the secondary transmitter with a constraint on the interference temperature (IT) specified by the PUs. The outage probability of the secondary receiver is explicitly derived in an exact closed-form expression. Also, some special cases of practical interest, including co-located PUs and massive MIMO, are discussed. Further, to mitigate instantaneous excessive interference onto PUs caused by the time-average IT, an iterative antenna reduction algorithm is developed for the secondary transmitter and, accordingly, the average number of transmit antennas is analytically computed. Extensive numerical and simulation results corroborate the effectiveness of our analysis. 
\end{abstract}

\begin{IEEEkeywords}
Cognitive radio (CR), interference, multiple-input multiple-output (MIMO), optimal power optimization, spatial multiplexing, zero-forcing (ZF) detection.
\end{IEEEkeywords}

\IEEEpeerreviewmaketitle

\section{Introduction}
\IEEEPARstart{C}{ognitive} radio (CR) is widely recognized as a promising technique to resolve the issue of spectrum scarcity, caused by the explosive growth of wireless data traffic. Among the three major paradigms to deploy CR in practice (i.e., interweave, overlay and underlay), underlay CR allows simultaneous transmissions of primary users (PUs) and secondary users (SUs), as well as low implementation complexity \cite{XiaCM13}. In a practical underlay CR system, to guarantee the quality of service (QoS) of PUs granted spectrum resources, the transmit (Tx) power of SUs with no fixed spectrum resources is strictly limited, such that the harmful interference from SUs to PUs remains below a prescribed tolerable level. To improve the performance of secondary transmission, multiple-input multiple-output (MIMO) and even massive MIMO antenna techniques can be explored since they provide additional degrees of freedom (DoF) in spatial domain, compared with traditional single-input single-output (SISO) transmission. When MIMO antenna was integrated into underlay CR systems, the spatial diversity gain of MIMO was widely used to enhance the reliability of secondary transmission, see e.g., \cite{j:Sarvendranath2013, c:XiongMukherjee2015} and references therein. On the other hand, the spatial multiplexing gain of MIMO was exploited to improve the data rate of secondary transmission, see e.g., \cite{c:YangQaraqe2013}. 

To avoid harmful interference onto PUs, various beamforming and/or power control strategies were proposed for secondary transmitters (STs). For instance, in \cite{j:NoamGoldsmith2013} the beamforming strategy of ST was carefully designed such that the SU transmits in the null space of the interference channel to its nearby PU. On the other hand, due to extreme difficulty to acquire perfect channel state information (CSI) in cognitive underlay systems where PUs are generally reluctant to cooperate with SUs, most power control strategies resorted to second-order CSI statistics, e.g., time-average channel gains. For instance, in \cite{c:RopokisBerberidis} a limited feedback solution was provided to optimal power allocation while in \cite{j:GopalakrishnanSidiropoulos2015} only binary and infrequent CSI was exploited to design beamforming algorithm. Interestingly, it was demonstrated in \cite{j:TourkiKhan2014} that the average CSI based power allocation strategy outperforms that based on instantaneous CSI, given that a low interference power constraint is dictated by PUs. However, if the time-average CSI pertaining to the channel from ST to primary receiver (PR) is exploited when optimal power allocation is performed at ST, it will inevitably introduce unexpected instantaneous excessive interference to PUs. In other words, instantaneous CSI-based power allocation strategy can guarantee that PUs are always free of excessive interference (i.e., the real interference is always no larger than the prescribed tolerable interference), whereas the average CSI-based power allocation strategy cannot. 

In the aforementioned works, neither the detrimental effect of inter-system interference nor a power allocation optimization of the secondary transmission were considered. In this paper, capitalizing on the latter observation, an underlay MIMO CR system is studied, where ST sends multiple parallel data streams to its corresponding receiver equipped with a zero-forcing (ZF) detector, under the constraint of tolerable interference power dictated by PUs. As well-known, ZF detector manifests a good performance-complexity tradeoff \cite{j:MiridakisSurvey}, compared with minimum mean-squared error (MMSE) detector or alike. Meanwhile, several single-antenna PUs are active and randomly located in the vicinity of the secondary users (the scenario of co-located multiple-antenna PUs is also considered as a special case). It is assumed that the secondary receiver (SR) is aware of perfect CSI only between itself and ST whereas the time-average channel gains between ST and its adjacent PR are available at SR. In summary, four major contributions of this work are as follows.
\begin{itemize}
	\item A new simple power allocation scheme is designed for ST in underlay MIMO CR systems. 
	\item Outage probability of the secondary transmission is explicitly derived, under independent Rayleigh fading channels, with extensive discussions of some special cases of practical interest, namely, co-located PUs and massive MIMO.
	\item To mitigate excessive interference onto PUs, an iterative antenna reduction algorithm is developed. 
	\item Based on the developed algorithm, the average number of active secondary Tx antennas is explicitly computed. 
\end{itemize}

To detail the aforementioned contributions, the rest of this paper is organized as follows. Section \ref{System Model} describes the system model. Section \ref{Optimal Power Allocation of the Secondary System} devises an optimal power allocation at ST. Section \ref{Performance Analysis of the Secondary System} analyzes the outage probability of the secondary transmission. Also, several special cases of practical interest are discussed. Next, to mitigate unexpected excessive interference to PUs, Section \ref{secimpact} develops an iterative antenna reduction algorithm. Afterwards, Section \ref{Numerical Results} presents simulation results compared with numerical ones, while Section \ref{Conclusion} concludes the paper. Some detailed derivations are relegated to Appendix.

{\bf Notation}: Vectors and matrices are denoted by lowercase and uppercase bold symbols (e.g., $\mathbf{x}$ and $\mathbf{X}$), respectively. The superscripts $(\cdot)^{-1}$,  $(\cdot)^{\dagger}$, $(\cdot)^{T}$ and $(\cdot)^{\mathcal{H}}$ means the inverse, pseudo-inverse, transpose and conjugate transpose, respectively. $\Tr[\mathbf{X}]$ calculates the trace of $\mathbf{X}$. $\mathbf{x}_{i}$ denotes the $i^{\rm th}$ entry of $\mathbf{x}$ while $[\mathbf{X}]_{ij}$ stands for the $(i, j)$ element of $\mathbf{X}$. $\diag\{x_{i}\}^{n}_{i=1}$ means a diagonal matrix with entries $x_{1}, \cdots, x_{n}$.  The operator $(x)^{+}$ equals $x$ if $x>0$, and zero otherwise. $|x|$ takes the absolute value of $x$ while $\|\mathbf{x}\|$ is the Euclidean norm of $\mathbf{x}$. $\mathbf{I}_{v}$ stands for the identity matrix of size $v \times v$. $\mathbb{E}[\cdot]$ is the expectation operator. The symbol $\overset{\text{d}}=$ means equality in distribution. The functions $f_{X}(\cdot)$, $F_{X}(\cdot)$ and $\overline{F}_{X}(\cdot)$ represent probability density function (PDF), cumulative distribution function (CDF) and complementary CDF (CCDF) of a random variable (RV) $X$, respectively. Complex-valued Gaussian RVs with mean $\mu$ and variance $\sigma^{2}$ is denoted as $\mathcal{CN}(\mu,\sigma^{2})$ while central chi-squared RVs with $v$ DoF as $\mathcal{X}^{2}_{v}$. Finally, $\Gamma(\cdot)$ denotes the Gamma function \cite[Eq. (8.310.1)]{tables} and $\Gamma(\cdot,\cdot)$ is the upper incomplete Gamma function \cite[Eq. (8.350.2)]{tables}.

\section{System Model}
\label{System Model}
As illustrated in Fig.~\ref{fig1}, we investigate an underlay MIMO communication system in the context of CR, where $M$ and $N$  antennas are equipped at ST and SR, respectively, with $N\geq M$. The Tx antennas operate in a spatial multiplexing mode and $M$ independent data streams are simultaneously transmitted in a given time instance. ZF detection is adopted at the receiver side. On the other hand, there are $L_{T}$ primary transmitters (PT) communicating with $L_{R}$ PRs, each with a single Tx/Rx antenna. Independent Rayleigh channel fading conditions are assumed for all the involved links. Also, both ST and SR are able to acquire statistical (second-order) CSI with respect to the channel gains of the primary system,\footnote{In principle, CSI of the links between the primary and secondary nodes can be obtained through a feedback channel from the primary service or via a band manager that mediates the exchange of information between the primary and secondary networks \cite{XiaCM13,j:HongBan2012}.} while perfect CSI is assumed regarding the channel gains between ST and SR.

\begin{figure}[!t]
\centering
\includegraphics[keepaspectratio,width=2.8in]{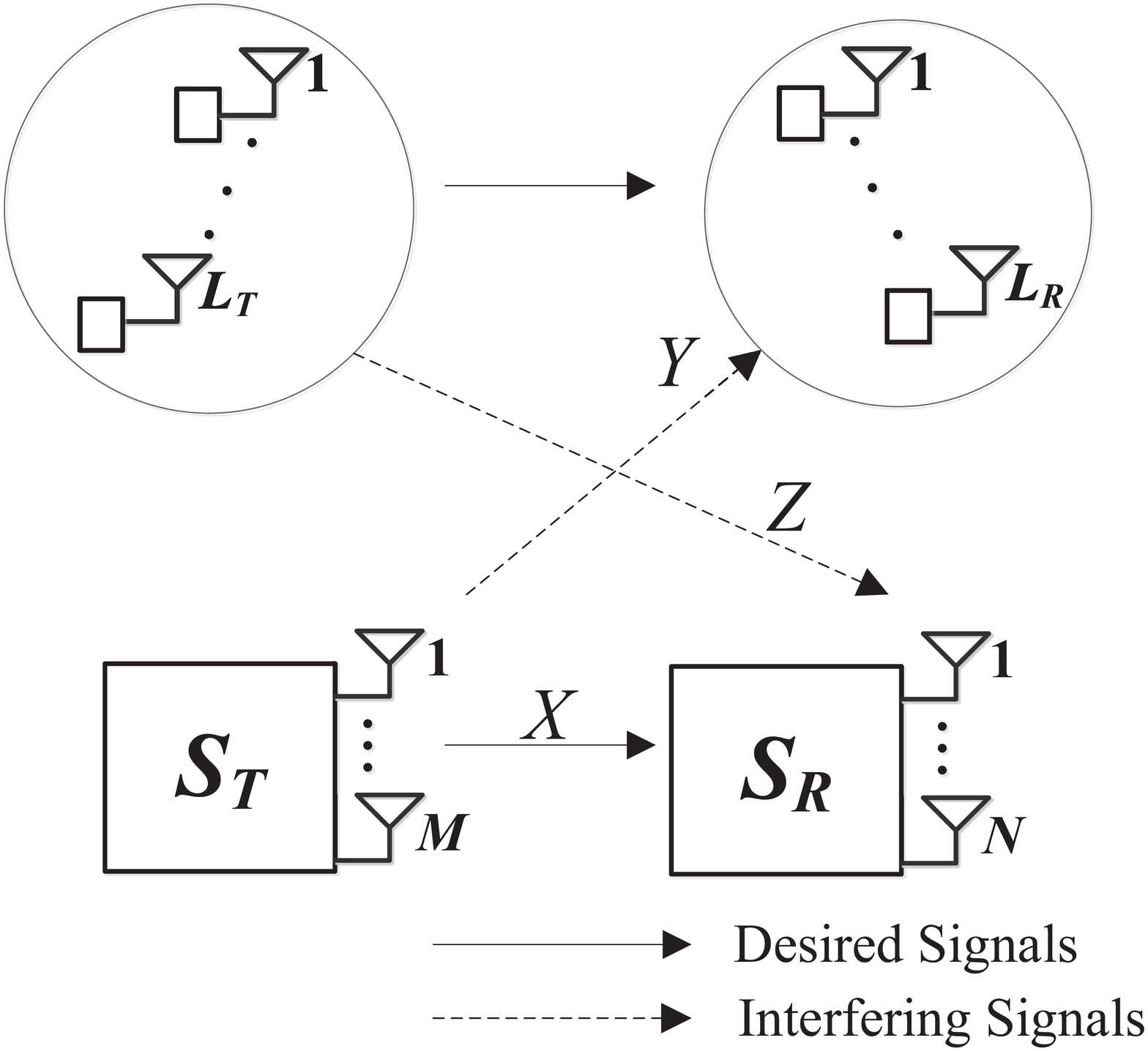}
\caption{The considered system configuration with $L_{T}$ PT and $L_{R}$ PRs, where $S_{T}$ and $S_{R}$ stand for the secondary transmitter and receiver, respectively. The parameters $X$, $Y$ and $Z$ denote the involved channel gains between the primary and secondary system, which are explicitly defined in Section \ref{Optimal Power Allocation of the Secondary System}.}
\label{fig1}
\end{figure}

The received signal at SR is given by
\begin{equation}
\mathbf{r} 
= \mathbf{H}\mathbf{P}^{\frac{1}{2}}\mathbf{s}+\sqrt{p_{\rm p}}\mathbf{H}_{\text{p}}\mathbf{s}_{\text{p}}+\mathbf{w},
\label{eq1}
\end{equation}
where $\mathbf{H}\in \mathbb{C}^{N\times M}$ denotes the desired channel from ST to SR, $\mathbf{s} \in \mathbb{C}^{M\times 1}$ represents the transmitted signals from the secondary source, $\mathbf{H}_{\text{p}}\in \mathbb{C}^{N\times L_{T}}$ stands for the interfering channel from PTs to SR, $\mathbf{s}_{\text{p}} \in \mathbb{C}^{L_{T}\times 1}$ is the transmitted signals from PTs, and $\mathbf{w} \in \mathbb{C}^{N\times 1}$ models the additive white Gaussian noise (AWGN) at SR. Moreover, $\mathbf{P} \in \mathbb{R}^{M\times M} = \diag\{p_{i}\}^{M}_{i=1}$ is a diagonal matrix with $p_i$ being the optimal Tx power at the $i^{\rm th}$ antenna of ST (to be explicitly determined in Section~\ref{PowerAllocation}), $p_{\rm p}$ is a constant used to denote the fixed Tx power at each PT, and $\mathbf{w}\overset{\rm d} = \mathcal{CN}(\mathbf{0}, N_{0}\mathbf{I}_{N})$ with $N_{0}$ being the AWGN variance. Without loss of generality, the Tx power of signals at either secondary or PTs are normalized, i.e., $\mathbb{E}[\mathbf{s}\mathbf{s}^{\mathcal{H}}] = \mathbf{I}_{M}$ and $\mathbb{E}[\mathbf{s}_{\text{p}}\mathbf{s}_{\text{p}}^{\mathcal{H}}] = \mathbf{I}_{L_{T}}$.

Based on the principle of ZF detection, an estimation of the transmitted symbol vector can be written as
\begin{equation}
\mathbf{r}^{\prime} 
 \triangleq \mathbf{G}^{\dagger} \mathbf{r}
 = \mathbf{s}+\sqrt{p_{\rm p}}\mathbf{G}^{\dagger}\mathbf{H}_{\text{p}}\mathbf{s}_{\text{p}}+\mathbf{G}^{\dagger}\mathbf{w},
\label{rzf}
\end{equation}
where 
\begin{equation}
\mathbf{G}^{\dagger} 
\triangleq (\mathbf{H}\mathbf{P}^{\frac{1}{2}})^{\dagger}
 = \left((\mathbf{H}\mathbf{P}^{\frac{1}{2}})^{\mathcal{H}}\mathbf{H}\mathbf{P}^{\frac{1}{2}}\right)^{-1}(\mathbf{H}\mathbf{P}^{\frac{1}{2}})^{\mathcal{H}}.
\end{equation}

\begin{remark}
It is noteworthy that Fig.~\ref{fig1} illurstrates a general case of primary transmission, where $L_{T}$ PTs and $L_{R}$ receivers are scattered and operate like a distributed MIMO system. This scenario includes the typical MIMO link with $L_{T}$ Tx and $L_{R}$ Rx antennas as a special case. The latter case is analyzed in Lemma 3 and Corollary 2.
\end{remark}

\section{Optimal Power Allocation at the Secondary Transmitter}
\label{Optimal Power Allocation of the Secondary System}
In this section, we start with formulating the signal to interference plus noise ratio (SINR) of each received secondary stream. Then, the optimal power allocation at ST is analytically presented.

\subsection{SINR Analysis}
After performing ZF detection as per \eqref{rzf}, the received SINR of the $i^{\rm th}$ secondary transmitted data stream, is given by
\begin{equation}
{\rm SINR}_{i} 
= \frac{1}{p_{\rm p}\left\|\left[\mathbf{G}^{\dagger}\right]_{i}\mathbf{H}_{\text{p}}\right\|^{2} + N_{0}\left\|\left[\mathbf{G}^{\dagger}\right]_{i}\right\|^{2}},
\label{sinri}
\end{equation}
where $[\mathbf{G}^{\dagger}]_{i}$ denotes the $i^{\rm th}$ row of $\mathbf{G}^{\dagger}$. Due to the high complexity of \eqref{sinri}, an exact analysis of ${\rm SINR}_{i} $ is almost mathematically intractable. For ease of further proceeding, the following Lemma~\ref{Lemma-1} reformulates \eqref{sinri} in a more tractable way. 
\begin{lemma}
\label{Lemma-1}
The received SINR of the $i^{\rm th}$ secondary transmitted data stream given by \eqref{sinri} is distributed as
\begin{equation}
{\rm SINR}_{i}\overset{\rm d} 
= \frac{p_{i}X_{i}}{p_{\rm p}Z+N_{0}},
\label{sinridistr}
\end{equation} 
where $X_{i}$ and $Z$ are mutually independent RVs. Moreover,
\begin{equation}
X_{i} \triangleq \sum^{N-M+1}_{l=1}|h^{(i)}_{l}|^{2},
\label{xidistr}
\end{equation} 
with $|h^{(i)}_{l}|^{2}$ being the channel gain from the $i^{\rm th}$ secondary Tx antenna to the $l^{\rm th}$ secondary Rx antenna, and
\begin{equation}
Z \triangleq p_{\rm p}\sum^{L_{T}}_{j=1}|z_{j}|^{2}, 
\label{zdistr}
\end{equation}
with $|z_{j}|^{2}$, $\forall j\in [1, L_{T}]$, are  independent and non-identically distributed (i.n.i.d.) exponential RVs.
\end{lemma}

\begin{IEEEproof}
The proof is relegated in Appendix \ref{appSINRdistr}.
\end{IEEEproof}

\subsection{Power Allocation}
\label{PowerAllocation}
Since multiple Tx antennas of ST operate in the spatial multiplexing manner, maximizing the data rate of secondary transmission is equivalent to maximizing the achievable data rate of each secondary data stream. In turn, this can be achieved by proportionally maximizing the corresponding Tx power at each antenna. Yet, this Tx power should not exceed a predefined threshold dictated by primary users, which is widely known as \emph{interference temperature} (IT). Since only statistical CSI of secondary-to-primary links (and vice-versa given that channel reciprocity is assumed) is available, an average IT threshold, $Q$, such that $\mathbb{E}[\sum^{M}_{i=1}p_{i}|y_{i}|^{2}] \leq Q, \ \forall p_{i}$ should be satisfied, while it holds that
\begin{equation}
|y_{i}|^{2} \triangleq \max_{j}\left\{\left|y^{(i)}_{j}\right|^{2}\right\}^{L_{R}}_{j=1},
\label{yii}
\end{equation}
where $|y^{(i)}_{j}|^{2}$ denotes the instantaneous interfering power caused by the $i^{\rm th}$ Tx antenna of ST to the $j^\text{th}$ PR. Since the distance between two Tx antennas at ST is negligible as compared with the distance between ST and any PR, we have that $\mathbb{E}[Y_{1}] =\mathbb{E}[Y_{2}] = \cdots = \mathbb{E}[Y_{M}] \triangleq \mathbb{E}[Y]$, where $\mathbb{E}[Y] \triangleq \mathbb{E}[|y_{i}|^{2}]$ with $|y_{i}|^{2}$ defined in \eqref{yii}. Thereby, by invoking the property of linearity for expected values, the aforementioned constraint with respect to IT becomes $\mathbb{E}[\sum^{M}_{i=1}p_{i}|y_{i}|^{2}] \leq Q\Rightarrow \sum^{M}_{i=1}p_{i}\mathbb{E}[Y] \leq Q$. 

On the other hand, it is more likely that primary nodes keep arbitrary distances from ST (cf. Fig.~\ref{fig1}), which in turn yields non-identical link statistics from different primary nodes to ST. 

Consequently, the optimization problem of Tx power allocation for ST can be formulated as
\begin{subequations}
\begin{align}
\mathcal{P}_{1}: \ &\max \sum^{M}_{i=1}\log_{2}\left(1+\frac{p_{i}X_{i}}{p_{\rm p}\mathbb{E}[Z]+N_{0}}\right)  \label{Opt_2} \\
&\ \text{s.t.}\ \ \sum^{M}_{i=1}p_{i} \, \mathbb{E}\left[Y\right]\leq Q, \quad \forall p_{i} \geq 0.
\label{optproblem1}
\end{align}
\end{subequations}
Obviously, the above formulation is suboptimal since only $\mathbb{E}[Z]$ is exploited by ST instead of its instantaneous counterpart. Nonetheless, this approach is suitable for practical applications as it is usually hard and/or costly to obtain instantaneous CSI in real-life wireless networks. 

To facilitate the subsequent analysis, the following Lemmas \ref{Lemma-2} and \ref{Lemma-3} formalize some key results regarding the statistics of $Y$, corresponding to different channel fading characteristics. 

\begin{lemma}
\label{Lemma-2}
In the case of i.n.i.d. fading channels, $\mathbb{E}[Y]$ can be analytically given by
\begin{eqnarray}
\mathbb{E}[Y]
& = &\sum^{L_{R}}_{l=1}\sum^{L_{R}}_{k=0}\underbrace{\sum^{L_{R}}_{n_{1}=1}\cdots \sum^{L_{R}}_{n_{k}=1}}_{n_{1}\neq \cdots \neq n_{k}\neq k}\frac{(-1)^{k}}{k!}  \nonumber \\
&     &{}\times \frac{1}{\mathbb{E}[Y^{(l)}]\left(\frac{1}{\mathbb{E}[Y^{(l)}]}+\sum\limits^{k}_{t=1}\frac{1}{\mathbb{E}[Y^{(n_{t})}]}\right)^{2}},
\label{EY}
\end{eqnarray}
where $\mathbb{E}[Y^{(l)}]$ denotes the average interfering channel gain from ST to the $l^\text{th}$, $\forall l \in [1, L_{R}]$, PR. In addition, the average total interfering channel gain from PTs to SR, i.e., $\mathbb{E}[Z]$, stems as
\begin{equation}
\mathbb{E}[Z] 
= \sum^{L_{T}}_{k=1}\left(\prod^{k}_{j=1,j\neq k}\frac{\mathbb{E}[Z_{k}]}{\mathbb{E}[Z_{k}]-\mathbb{E}[Z_{j}]}\right)\mathbb{E}[Z_{k}],
\label{EZ}
\end{equation}
where $\mathbb{E}[Z_{l}]$ denotes the average interfering channel gain from the $l^\text{th}$, $\forall l \in [1, L_{T}]$, PT to SR.
\end{lemma}

\begin{proof}
Please refer to Appendix \ref{appEYderiv}.
\end{proof}

On the other hand, in the case when co-located primary Tx/Rx antennas are considered, e.g., a typical MIMO transceiver, the total interfering power between the primary and secondary system can be efficiently modeled by independent and identically distributed (i.i.d.) RVs, and we have the following lemma.
\begin{lemma}
\label{Lemma-3}
For the scenario of i.i.d. interfering channels, it stems that
\begin{equation}
\mathbb{E}[Y]
= L_{R}\mathbb{E}[Y^{(\text{i.i.d.})}]\sum^{L_{R}-1}_{k=0}\frac{(-1)^{k}}{(k+1)^{2}}\binom{L_{R}-1}{k},
\label{EYiid}
\end{equation}
where $\mathbb{E}[Y^{(\text{i.i.d.})}]$ denotes the average (identical) interfering channel gain from ST to PR. Further, we get
\begin{align}
\mathbb{E}[Z]=L_{T}\mathbb{E}[Z_{\text{i.i.d.}}],
\label{EZiid}
\end{align} 
where $\mathbb{E}[Z_{\text{i.i.d.}}]$ stands for the average (identical) channel gain from each PT to SR. 
\end{lemma}

\begin{proof}
In the case of co-located primary Tx antennas, $f_{Y}(\cdot)$ is formed by the maximum of $L_{R}$ i.i.d. exponential RVs, which is expressed as
\begin{equation}
f_{Y}(y) 
= \sum^{L_{R}-1}_{k=0}\frac{(-1)^{k}\binom{L_{R}-1}{k}L_{R}}{\mathbb{E}[Y^{(\text{i.i.d.})}]}\exp\left(-\frac{(k+1)y}{\mathbb{E}[Y^{(\text{i.i.d.})}]}\right).
\label{fYiid}
\end{equation} 
Thereby, inserting (\ref{fYiid}) into (\ref{EEEY}) and performing some algebraic manipulations yields \eqref{EYiid}.

Since the distribution of the total interfering power from the primary to secondary system can be obtained by an $L_{T}$-fold convolution of i.i.d. exponential RVs, it yields a $\mathcal{\chi}^{2}_{2 L_{T}}$ distribution. Therefore, \eqref{EZiid} can be readily obtained.
\end{proof}

Now, we are in a position of performing optimal power allocation at ST. Clearly, the constraint given by \eqref{optproblem1} consists of a linear sum and, thus, typical convex optimization solver can be applied to $\mathcal{P}'_{1}$, yielding the optimal Tx power for the $i^{\rm th}$ data stream as
\begin{align}
p^{\star}_{i}=\left(\frac{\lambda}{\text{ln}(2)\mathbb{E}[Y]}-\frac{p_{\rm p}\mathbb{E}[Z]+N_{0}}{X_{i}}\right)^{+},
\label{optpow}
\end{align}
where $\lambda$ denotes the associated Lagrangian multiplier. Notice that $\lambda$ is a common parameter used for all the simultaneously transmitted data streams due to the average IT constraint and the statistically identical channel fading conditions specified before \eqref{Opt_2}. In particular, the value of $\lambda$ can be efficiently calculated by turning the inequality of \eqref{optproblem1} into equality, such that
\begin{eqnarray}
\lefteqn{\sum^{M}_{i=1}p_{i}\mathbb{E}\left[Y\right] = Q} \nonumber \\
& \Longrightarrow & \mathbb{E}_{X_{i},|y_{i}|^{2}}\left[\sum^{M}_{i=1}p_{i}|y_{i}|^{2}\right] = Q  \nonumber \\
& \Longrightarrow & \sum^{M}_{i=1}\mathbb{E}_{X_{i}}\left[p_{i}\right] = \frac{Q}{\mathbb{E}[Y]} \nonumber \\
& \Longrightarrow & \mathbb{E}_{X_{i}}[p_{i}]=\frac{Q}{M\mathbb{E}[Y]}, \, \forall i\in[1, M],
\label{langrange}
\end{eqnarray}
where the fact that $X_{i}$ and $|y_{i}|^{2}$, $\forall i\in[1, M]$, are mutually independent is used. More specifically, we have the following lemma.

\begin{lemma}
The value of Lagrangian multiplier, $\lambda$, can be obtained by numerically solving the following expression with respect to $\lambda$,
\begin{align}
\left\{
\begin{array}{l l} 
\frac{\mathbb{E}[X]}{\Gamma(N-M+1)}\Bigg[\frac{\lambda \mathbb{E}[X]\Gamma\left(N-M+1,\frac{\mathcal{C}}{\mathbb{E}[X]}\right)}{\ln{(2)}\mathbb{E}[Y]}\\
-\left(p_{\rm p}\mathbb{E}[Z]+N_{0}\right)\Gamma\left(N-M,\frac{\mathcal{C}}{\mathbb{E}[X]}\right)\Bigg]&\\
= \min\left\{\frac{Q}{M\mathbb{E}[Y]},\frac{p_{\max}}{M}\right\}, &{\rm for }\ N>M,\\
& \\
\frac{\exp\left(-\frac{\mathcal{C}}{\mathbb{E}[X]}\right)\lambda}{\ln{(2)} \, \mathbb{E}[Y]}-\frac{\left(p_{\rm p}\mathbb{E}[Z]+N_{0}\right)\Gamma\left(0,\frac{\mathcal{C}}{\mathbb{E}[X]}\right)}{\mathbb{E}[X]}\\
= \min\left\{\frac{Q}{M\mathbb{E}[Y]},\frac{p_{\max}}{M}\right\}, &{\rm for }\ N=M, 
\end{array}\right. 
\label{langrangecl}
\end{align}
where $p_{\max}$ is the maximum allowable power at ST (implying that $p_{\max}/M$ is the maximum allowable power at each Tx antenna)\footnote{{\color{black}In current work, it is assumed that the hardware gear of each secondary node is identical; the RF
circuit, power amplifier, etc., all have the same hardware characteristics such
that the maximum achievable power for each Tx antenna reaches up to $p_{\max}/M$.}} and
\begin{equation}
\mathcal{C} 
\triangleq \frac{\ln{(2)}}{\lambda}\left(\mathbb{E}[Y]N_{0}+p_{\rm p}\mathbb{E}[Y] \, \mathbb{E}[Z]\right).
\label{c}
\end{equation}
\end{lemma}

\begin{proof}
Please refer to Appendix \ref{applagrange}.
\end{proof}

\begin{remark}
It is remarkable that since all the involved average channel gains, i.e., $\mathbb{E}[X]$, $\mathbb{E}[Y]$ and $\mathbb{E}[Z]$, are analytically provided, the value of $\lambda$ can be efficiently computed in polynomial time, by using standard numerical-solving methods. Furthermore, it is observed from (\ref{optpow}) that $X_{i}>\ln{(2)}\mathbb{E}[Y]\left(p_{\rm p}\mathbb{E}[Z]+N_{0}\right)/\lambda$ should hold in order to activate the transmission of the $i^{\rm th}$ secondary data stream, which implies that the minimum channel gain needed to activate this data stream is proportional to $\mathbb{E}[Y]$ and $\mathbb{E}[Z]$, yet inversely proportional to the Lagrangian multiplier $\lambda$.
\end{remark}

\section{Performance Analysis of the Secondary Transmission}
\label{Performance Analysis of the Secondary System}
In this section, outage performance of the secondary system is derived in an exact closed-form expression, while some special cases of practical interest are discussed. 

\subsection{Outage Probability of the Secondary Transmission}
By definition, outage probability of the $i^{\rm th}$ secondary data stream, $P^{(i)}_{\text{out}}(\gamma_{\text{th}})$, $\forall i \in [1, M]$, is the probability that its received SINR falls below a prescribed threshold value $\gamma_{\text{th}}\triangleq 2^{\mathcal{R}_{T}}-1$, where $\mathcal{R}_{T}$ stands for a target data rate in the unit of bps/Hz.

In light of \eqref{optpow}, the instantaneous received power of the $i^{\rm th}$ secondary data stream at SR is given by 
\begin{equation}
p^{\star}_{i}X_{i}
= \left(\frac{\lambda X_{i}}{\text{ln}(2)\mathbb{E}[Y]}-\left(p_{\rm p}\mathbb{E}[Z]+N_{0}\right)\right)^{+},
\label{snri}
\end{equation}
whose CDF is formalized in the following lemma.

\begin{lemma}
The CDF of the instantaneous received power of the $i^{\rm th}$ secondary stream, $p^{\star}_{i}X_{i}$, is presented as
\begin{align}
\nonumber
F_{p^{\star}_{i}X_{i}}(x)&=1-\exp\left(-\frac{\ln{(2)}\mathbb{E}[Y]x}{\lambda \mathbb{E}[X]}-\frac{\mathcal{C}}{\mathbb{E}[X]}\right)\\
&\times \sum^{N-M}_{l=0}\frac{\left(\frac{\ln{(2)}\mathbb{E}[Y]x}{\lambda \mathbb{E}[X]}+\frac{\mathcal{C}}{\mathbb{E}[X]}\right)^{l}}{l!}. 
\label{uncpdfsnr}
\end{align}
\end{lemma}

\begin{proof}
The proof is relegated in Appendix \ref{apppdfsnr}.
\end{proof}

Now we are in a position to formulate the outage probability of the $i^{\rm th}$ secondary data stream.
\begin{theorem}
Outage probability of the $i^{\rm th}$ secondary data stream can be expressed in closed-form as
\begin{align}
\nonumber 
&P^{(i)}_{\rm out}(\gamma_{\rm th})
 =1-\sum^{N-M}_{l=0}\sum^{L_{T}}_{k=1}\left(\prod^{k}_{j=1, j\neq k}\frac{\mathbb{E}[Z_{k}]}{\mathbb{E}[Z_{k}]-\mathbb{E}[Z_{j}]}\right)\\
\nonumber 
&\times \frac{\left(\frac{\ln{(2)}\mathbb{E}[Y]p_{\rm p}\gamma_{\rm th}}{\lambda \mathbb{E}[X]}\right)^{l}\exp\left(\frac{\left(\frac{\ln{(2)}\mathbb{E}[Y]N_{0}\gamma_{\rm th}}{\lambda \mathbb{E}[X]}+\frac{\mathcal{C}}{\mathbb{E}[X]}\right)}{\left(\frac{\ln{(2)}\mathbb{E}[Y]p_{\rm p}\mathbb{E}[Z_{k}]\gamma_{\rm th}}{\lambda \mathbb{E}[X]}\right)}\right)}{l!\mathbb{E}[Z_{k}]\left(\frac{\ln{(2)}\mathbb{E}[Y]p_{\rm p}\gamma_{\rm th}}{\lambda \mathbb{E}[X]}+\frac{1}{\mathbb{E}[Z_{k}]}\right)^{l+1}}\\
&\times \Gamma\left(l+1,\frac{\left(\frac{\ln{(2)}\mathbb{E}[Y]p_{\rm p}\gamma_{\rm th}}{\lambda \mathbb{E}[X]}+\frac{1}{\mathbb{E}[Z_{k}]}\right)\left(\frac{\ln{(2)}\mathbb{E}[Y]N_{0}\gamma_{\rm th}}{\lambda \mathbb{E}[X]}+\frac{\mathcal{C}}{\mathbb{E}[X]}\right)}{\left(\frac{\ln{(2)}\mathbb{E}[Y]p_{\rm p}\gamma_{\rm th}}{\lambda \mathbb{E}[X]}\right)}\right).
\label{poutcl}
\end{align}
\end{theorem} 

\begin{proof}
The proof is provided in Appendix \ref{apppoutcl}.
\end{proof}

Since $(l+1) \in \mathbb{N}^{+}$, the upper incomplete Gamma function $\Gamma(\cdot, \cdot)$ in \eqref{poutcl} can alternatively be expressed in terms of finite sum series of elementary functions \cite[Eq. 8.352.2]{tables}. That is, \eqref{poutcl} can be rewritten as finite sum series of elementary functions. Hence, the outage probability given by \eqref{poutcl} can be accurately and efficiently calculated by using popular numerical softwares, such as Matlab and Mathematica.

Besides the outage probability, other important system performance metrics can be readily obtained by means of a simple numerical integration of \eqref{poutcl}. For instance, the average ergodic capacity of the $i^{\rm th}$ data stream is given by $(1/\text{ln}(2))\int^{\infty}_{0}\left(1-P^{(i)}_{\rm out}(x)\right)/(1+x){\rm d}x$, while the average symbol-error rate for binary modulations is captured by $(\mathcal{A}\sqrt{\mathcal{B}}/(2 \sqrt{\pi}))\int^{\infty}_{0}\left(\exp(-\mathcal{B} x)/\sqrt{x}\right)P^{(i)}_{\rm out}(x){\rm d}x$ with $\mathcal{A}$ and $\mathcal{B}$ standing for fixed parameters determined by the used modulation scheme. 

\subsection{Special Cases of Practical Interest}
Now, we study two special cases of practical interest, and their respective outage probabilities are explicitly presented.

\underline{\emph{Case i}}: Equal Number of Antennas at Secondary Transmitter and Receiver.
\begin{corollary}
In the case of $M=N$, i.e., ST and SR have a same number of antennas, the outage probability given by \eqref{poutcl} reduces to 
\begin{align}
\nonumber
P^{(i)}_{\rm out}(\gamma_{\rm th})
 &=1-\sum^{L_{T}}_{k=1}\left(\prod^{k}_{j=1, j\neq k}\frac{\mathbb{E}[Z_{k}]}{\mathbb{E}[Z_{k}]-\mathbb{E}[Z_{j}]}\right)\\
&\times \frac{\exp\left(-\frac{\ln{(2)}\mathbb{E}[Y]N_{0}\gamma_{\rm th}}{\lambda \mathbb{E}[X]}-\frac{\mathcal{C}}{\mathbb{E}[X]}\right)}{\left(\frac{\ln{(2)}\mathbb{E}[Y]\mathbb{E}[Z_{k}]p_{\rm p}\gamma_{\rm th}}{\lambda \mathbb{E}[X]}+1\right)}. 
\label{poutcl11}
\end{align}
\end{corollary}
\begin{proof}
Using the equality $\Gamma(1, x) = \exp(-x)$ and performing some algebraic manipulations over \eqref{poutcl} yields \eqref{poutcl11} in a straightforward manner.
\end{proof}

\underline{\emph{Case ii}}: Identical Statistics of Primary-to-Secondary Interferences.
\begin{corollary}
For the case with $L_{T}$ co-located primary Tx antennas (e.g., a multi-antenna primary transmitter), the outage probability given by \eqref{poutcl} reduces to 
\begin{align}
\nonumber
P^{(i)}_{\rm out}(\gamma_{\rm th})&=1-\exp\left(-\frac{\ln{(2)}\mathbb{E}[Y]N_{0}\gamma_{\rm th}}{\lambda \mathbb{E}[X]}-\frac{\mathcal{C}}{\mathbb{E}[X]}\right)\\
\nonumber
&\times \sum^{N-M}_{l=0}\sum^{l}_{s=0}\binom{l}{s}\frac{s!}{l!(L_{T}-1)!\mathbb{E}[Z_{\rm i.i.d.}]^{L_{T}}}\\
&\times \frac{\left(\frac{\ln{(2)}\mathbb{E}[Y]N_{0}\gamma_{\rm th}}{\lambda \mathbb{E}[X]}+\frac{\mathcal{C}}{\mathbb{E}[X]}\right)^{l-s}\left(\frac{\ln{(2)}\mathbb{E}[Y]p_{\rm p}\gamma_{\rm th}}{\lambda \mathbb{E}[X]}\right)^{s}}{\left(\frac{\ln{(2)}\mathbb{E}[Y]p_{\rm p}\gamma_{\rm th}}{\lambda \mathbb{E}[X]}+\frac{1}{\mathbb{E}[Z_{\rm i.i.d.}]}\right)^{s+1}}.
\label{poutcliid}
\end{align}
\end{corollary}
\begin{proof}
By using a similar approach as that in Appendix~\ref{apppoutcl}, while noting that $Z\overset{\text{d}}=\mathcal{\chi}^{2}_{2 L_{T}}$, \eqref{poutcliid} can be readily derived.
\end{proof}

Moreover, similar to \eqref{poutcl11}, if $M=N$, \eqref{poutcliid} can be further reduced to
\begin{align}
\nonumber
&P^{(i)}_{\rm out}(\gamma_{\rm th})=\\
&1-\frac{\exp\left(-\frac{\ln{(2)}\mathbb{E}[Y]N_{0}\gamma_{\rm th}}{\lambda \mathbb{E}[X]}-\frac{\mathcal{C}}{\mathbb{E}[X]}\right)}{(L_{T}-1)!\mathbb{E}[Z_{\rm i.i.d.}]^{L_{T}}\left(\frac{\ln{(2)}\mathbb{E}[Y]p_{\rm p}\gamma_{\rm th}}{\lambda \mathbb{E}[X]}+\frac{1}{\mathbb{E}[Z_{\rm i.i.d.}]}\right)}.
\label{poutcliidmn}
\end{align}

\subsection{Asymptotic Analysis of Large Number of Secondary Tx/Rx Antennas}
Currently, massive MIMO systems are widely recognized as a cornerstone of the forthcoming 5G wireless communications \cite{j:massive}. From the information-theoretic point of view, the concept of massive MIMO implies that the number of Tx/Rx antennas becomes large (e.g., tens or hundreds) in practice and even approaches infinity in theory, i.e., $M$ and/or $N \rightarrow +\infty$. The main benefit arising from adopting this approach is the so-called \emph{channel hardening} effect, i.e., small-scale fading tends to vanish. Moreover, if $N \gg M$, the achievable spatial DoF and spectral efficiency are significantly enhanced. Accordingly, in this part we investigate the asymptotic received SINR at SR in the sense of large number of secondary Tx/Rx antennas.

\underline{\emph{Case i}}: $N \rightarrow +\infty$, while  $L_{T}$ and $M$ remain finite.

This scenario corresponds to the case where the secondary receiver consists of a massive MIMO array in the presence of $L_T$ distributed single-antenna primary nodes or a conventional MIMO PT with $L_T$ Tx antennas and a finite number of secondary Tx antennas.

{\color{black}
\begin{corollary}
Let $N$ tend to $\infty$, the number of primary and secondary Tx antennas ($L_T$ and $M$, respectively) are finite. The received SINR of the $i^{\rm th}$ secondary data stream grows asymptotically without any restriction, such that
\begin{equation}
{\rm SINR}_{i} \rightarrow \frac{\left(\frac{\lambda \mathbb{E}[X]}{\ln(2)\mathbb{E}[Y]}\right) (N-M+1)}{p_{\rm p}Z+N_{0}}\rightarrow +\infty.
\label{sinridistrrr1}
\end{equation}
\end{corollary}
}

\begin{proof}
By introducing the auxiliary variable vector $\mathbf{b}\overset{\text{d}}=\mathcal{CN}(\mathbf{0},\mathbf{I}_{N-M+1})$, while based on (\ref{xidistr}) and invoking the central limit theorem, (\ref{snri}) becomes
\begin{align}
\nonumber
p^{\star}_{i}X_{i}&=\left(\frac{\lambda \mathbb{E}[X]}{\text{ln}(2)\mathbb{E}[Y]}\right)\mathbf{b}^{\mathcal{H}}\mathbf{b}\\
&\rightarrow \left(\frac{\lambda \mathbb{E}[X]}{\text{ln}(2)\mathbb{E}[Y]}\right) (N-M+1),~\text{as~} N\rightarrow +\infty.
\label{xdistrasy}
\end{align}
Clearly, the effective channel gain of each received data stream is greatly enhanced in this case. Also, since the denominator of \eqref{sinridistr} is bounded, the desired result is directly extracted.
\end{proof}

\underline{\emph{Case ii}}: $M$, $N$ and $L_{T}$ tend to $\infty$ while $N/M \triangleq \kappa < \infty$.

This scenario can be realized when both the secondary and primary systems are equipped with massive MIMO antenna arrays. Similar to \cite{j:MatthaiouZFMassiveMIMO2013}, $Z$ in (\ref{zdistr}) can be rewritten in a quadratic form, i.e.,
\begin{align}
Z\triangleq \mathbf{a}^{\mathcal{H}}\mathbf{D}\mathbf{a},
\label{zdistrasy}
\end{align}
where $\mathbf{a}\overset{\text{d}}=\mathcal{CN}(\mathbf{0},\mathbf{I}_{L_{T}})$ and $\mathbf{D}=\diag\{\mathbb{E}[Z_{k}]\}^{L_{T}}_{k=1}$.

\begin{corollary}
\label{cordeterm}
In case $M$, $N$ and $L_{T}$ tend to infinity while $N/M \triangleq \kappa < \infty$, the received SINR of the $i^{\rm th}$ secondary data stream approaches a constant, given by
\begin{eqnarray}
\nonumber
{\rm SINR}_{i}
& \rightarrow & \frac{(N-M+1)\frac{\lambda \mathbb{E}[X]}{\ln{(2)}\mathbb{E}[Y]}}{\frac{p_{\rm p}}{L_{T}}\Tr[\mathbf{D}]}  \\
&         =        & \frac{(N-M+1)\frac{\lambda \mathbb{E}[X]}{\ln{(2)}\mathbb{E}[Y]}}{\frac{p_{\rm p}}{L_{T}}\sum\limits^{L_{T}}_{k=1}\mathbb{E}[Z_{k}]}.
\label{sinridistr1}
\end{eqnarray}
\end{corollary}

\begin{proof} 
It is clear that $Z\triangleq \mathbf{a}^{\mathcal{H}}\mathbf{D}\mathbf{a}=\Tr[\mathbf{D}]/L_{T}$, as $L_{T}\rightarrow +\infty$ \cite[Lemma 4]{j:HoydisBrink}. Hence, substituting \eqref{xdistrasy} and \eqref{zdistrasy} into \eqref{sinridistr} and dividing both the numerator and denominator of \eqref{sinridistr} with $N-M+1$, yields the desired \eqref{sinridistr1}.
\end{proof}

Note that both \eqref{sinridistrrr1} and \eqref{sinridistr1} confirm the channel hardening effect (i.e., the small-scale fading coincides with its average). However, when $L_{T} \rightarrow +\infty$, the received SINR is bounded, whereas both the secondary Tx power and primary-to-secondary interference play a key role to the overall performance of secondary transmission. Only in the scenario when $N \gg M$ and $N \gg L_{T}$, \eqref{sinridistr1} grows infinitely, while it is reduced to \eqref{sinridistrrr1}.

\begin{corollary}
\label{optpowermassivemimo}
In the case when $N$ tends to infinity with arbitrary $M \geq 1$, the optimal Tx power for each secondary data stream converges to the following deterministic value
\begin{align}
p^{\star}_{i}\rightarrow \min\left\{\frac{Q}{M \mathbb{E}[Y]},\frac{p_{\max}}{M}\right\},\ \ \forall i.
\label{determpower}
\end{align}
\end{corollary}

\begin{proof}
Using the first equality of \eqref{langrangee} and assuming that $N\rightarrow \infty$, it holds that
\begin{align}
\nonumber
&p^{\star}_{i} = \min\left\{\frac{Q}{M\mathbb{E}[Y]},\frac{p_{\max}}{M}\right\}\\
\nonumber
\Longleftrightarrow&\left(\frac{\lambda}{\text{ln}(2)\mathbb{E}[Y]}-\frac{\left(p_{\rm p}\mathbb{E}[Z]+N_{0}\right)}{X_{i}}\right)= \min\left\{\frac{Q}{M\mathbb{E}[Y]},\frac{p_{\max}}{M}\right\}\\
\nonumber
\Longrightarrow & \frac{\lambda}{\text{ln}(2)\mathbb{E}[Y]} \rightarrow \min\left\{\frac{Q}{M\mathbb{E}[Y]},\frac{p_{\max}}{M}\right\},\ \ N\rightarrow \infty\\
\Longrightarrow & \lambda \rightarrow \min\left\{\frac{\ln{(2)} Q}{M},\frac{\ln{(2)}\mathbb{E}[Y]p_{\max}}{M}\right\}.
\label{asylambdaa}
\end{align}
Hence, by recalling \eqref{optpow} and using the latter expression, the desired result is obtained.
\end{proof}

\underline{\emph{Case iii}}: $M$ and $N \rightarrow +\infty$ while $N/M \triangleq \kappa < \infty$, and $L_{T}$ remains finite.

This scenario corresponds to the case where the secondary link consists of massive MIMO arrays in the presence of $L_T$ distributed single-antenna primary nodes or a conventional MIMO PT with $L_T$ Tx antennas.

{\color{black}
\begin{corollary}
Let $M$ and $N$ tend to $\infty$ while their ratio remains constant, i.e., $N/M \triangleq \kappa < \infty$, and the number of primary Tx antennas ($L_T$) is finite. The received SINR of the $i^{\rm th}$ secondary data stream approaches the following constant
\begin{align}
\nonumber
{\rm SINR}_{i} &\rightarrow \frac{\left(\frac{\lambda \mathbb{E}[X]}{\ln(2)\mathbb{E}[Y]}\right) (N-M+1)}{p_{\rm p}Z+N_{0}}\\
&\rightarrow \frac{\left(\frac{\min\left\{Q,\mathbb{E}[Y]p_{\max}\right\} \mathbb{E}[X]}{\mathbb{E}[Y]}\right) (\kappa-1)}{(p_{\rm p}Z+N_{0})}.
\label{sinridistrrrrr1}
\end{align}
\end{corollary}

\begin{proof}
The result in \eqref{sinridistrrrrr1} can be easily obtained by substituting \eqref{asylambdaa} in the first equality of \eqref{sinridistrrr1} and after performing some straightforward manipulations.
\end{proof}
}

\section{Mitigating Unexpected Excessive Interference to Primary Users}
\label{secimpact}
When cognitive transmission technique is deployed in practice, it is almost infeasible or not affordable for secondary system to obtain accurate instantaneous CSI between numerous secondary and primary nodes \cite{j:Zhang2009}. As a result, an average tolerable IT threshold, $Q$, is widely used at primary nodes, to maintain a prescribed link quality between primary nodes. Although IT is an effective performance measure to guarantee the transmission quality of primary users and to enhance the performance of secondary users, it may introduce excessive instantaneous interference to primary users. This phenomenon is widely known as \emph{interference leakage} \cite{j:AliTransmitPrecoding}. To mitigate the effect of interference leakage, in this section an iterative antenna reduction algorithm is proposed.

Mathematically, an event of interference leakage can be defined as
\begin{equation}
\mathcal{L}^{(M)}: \ \min_{j}\left\{\sum^{M}_{i=1}\left(p^{\star}_{i}|y^{(i)}_{j}|^{2}\right)>Q\right\}, \, \forall j \in [1, L_{R}],
\label{probinstant}
\end{equation} 
where the superscript ``$M$'' of $\mathcal{L}^{(M)}$ denotes the actual number of secondary Tx antennas, and $|y^{(i)}_{j}|^{2}$ stands for the instantaneous interfering power caused by the $i^{\rm th}$ antenna of ST to the $j^\text{th}$ PR. Due to the high complexity of \eqref{probinstant}, performing an exact analysis of the event of interference leakage is extremely difficult, if not impossible. Thus, in the following, the statistics associated with $\mathcal{L}^{(M)}$ are evaluated and then used for later algorithm development.

\begin{corollary} 
Given the channel gains between each secondary Tx antenna and a single-antenna receiver, i.e., $|h_{i}|^{2}$, $\forall i \in [1, M]$, the probability of interference leakage can be expressed as
\begin{align}
& \lefteqn{{\rm Pr}\left[\mathcal{L}^{(M)} > Q\biggr\rvert \left\{|h_{i}|^{2}\right\}^{M}_{i=1}\right]}  \nonumber \\
& = \prod^{L_{R}}_{j=1}\left[\sum^{M}_{i=1}\left(\prod^{i}_{k=1, k \neq i}\frac{p^{\star}_{i}}{p^{\star}_{i}-p^{\star}_{k}}\right) 
\exp\left(\frac{-Q}{p^{\star}_{i} \, \mathbb{E}\left[\left|Y^{(i)}_{j}\right|^{2}\right]}\right)\right].
\label{probinstant1}
\end{align} 
\end{corollary}

\begin{proof}
It is evident that the summand term in \eqref{probinstant} consists of an $M$-sum series of i.n.i.d. exponential RVs, given $|h_{i}|^{2}$, $\forall i \in [1, M]$. By using a similar approach to that used to derive \eqref{fZ} and recalling the fact that the CCDF of the minimum of $L$ independent RVs is the $L$-product of CCDFs of each individual RV, \eqref{probinstant1} can be readily obtained. 
\end{proof}

\subsection{Proposed Iterative Antenna Reduction Algorithm}
To mitigate unexpected excessive interference from secondary to primary users, an iterative antenna reduction algorithm is developed in this part to deactivate some secondary Tx antennas in the case when the probability of interference leakage (computed by \eqref{probinstant1}) is higher than a prescribed threshold, say, $\mathcal{T}_{G}$. Specifically, this process includes four steps.
\begin{enumerate}
	\item With the prescribed interference threshold $Q > 0$ and interference leakage threshold $\mathcal{T}_{G} > 0$, ST initially uses its all $M$ Tx antennas to radiate signals, and calculates the probability of interference leakage as per \eqref{probinstant1}.
	\item If the computed probability of interference leakage is no larger than $\mathcal{T}_{G}$, ST continues to transmit over all Tx antennas.
	\item Otherwise, ST removes the Tx antenna that causes the highest average interference to the primary system and uses the remaining $M-1$ antennas for subsequent data transmission. Moreover, it calculates again the probability of interference leakage according to \eqref{probinstant1}.
	\item Steps 2) and 3) are repeated until no Tx antenna is available. In such a case, the secondary transmission has to be suspended until the next 
transmission slot arrives.
\end{enumerate}

It is noteworthy that in the proposed approach a fixed on/off strategy is applied at secondary Tx antennas. This strategy is suboptimal yet practical in real-life wireless networks where statistical CSI, rather than instantaneous CSI, are easier to acquire. In contrast, if instantaneous CSI is known, a more effective approach is to first turn off the Tx antenna which introduces the highest instantaneous interference to primary users. However, the latter approach is beyond the scope of this work.

For completeness of exposition, the proposed iterative antenna reduction approach is formalized in Algorithm~1.

\begin{algorithm}[t]
	\caption{Mitigating Unexpected Excessive Interference from Secondary to Primary Users}
	\begin{algorithmic}[1]
		 \INPUT{$M$, $Q$ and $\mathcal{T}_{G}$}
		 \OUTPUT{$M_{E}$ (the effective number of secondary Tx antennas)}
		  \WHILE{$M>0$}
			\STATE{Computing the probability of interference leakage as per Eq. \eqref{probinstant1}}
			\IF{The obtained value $\leq \mathcal{T}_{G}$}
				\STATE $M_{E} = M$;
				\STATE End of the algorithm;
			\ELSE $\:\:M=M-1$, such that the Tx antenna indicating $\max_{i}\mathbb{E}[|y^{(i)}_{j}|^{2}],\ 1\leq i\leq M,1\leq j\leq L_{R}$ is dropped.
				\STATE Go to Step 2;
			\ENDIF
		\ENDWHILE 
	\end{algorithmic}
\end{algorithm}

\subsection{Average Number of Active Secondary Tx Antennas}
According to the proposed Algorithm 1, the effective number of secondary Tx antennas, i.e., the number of active secondary Tx antennas, will change from time to time. Consequently, the average number of active secondary Tx antennas will be a valuable measure to account for the efficiency of secondary Tx antennas. In such a case, an estimation of the average number of active secondary Tx antennas will benefit system designers determining the real number of secondary Tx antennas.

Define $\mathcal{L}^{(l)}\triangleq \min_{j}\left\{\sum^{l}_{i=1}(p^{\star}_{i}\left|y^{(i)}_{j}\right|^{2}) > Q\right\}$. By recalling the proposed antenna reduction algorithm, three different cases will happen. More specifically, 
\begin{itemize}
	\item All the secondary Tx antennas are active if and only if ({\it iff}) $\text{Pr}[\mathcal{L}^{(M)} > Q] \leq \mathcal{T}_{G}$, which is equivalent to $\text{Pr}[\mathcal{L}^{(M)}\leq Q] > \mathcal{T}_{G}$. Therefore, the probability that such a case will happen is $F_{\mathcal{L}^{(M)}}(Q)\mathcal{U}\{\mathcal{T}_{G}\}$, where $\mathcal{U}\{y\}$ is the unitary step function such that $x\:\mathcal{U}\{y\}=x$ for $x>y$, while $x\:\mathcal{U}\{y\}=0$ for $x\leq y$.
	\item Given $1\leq l\leq M-1$, $l$ secondary Tx antennas will be activated {\it iff} $\text{Pr}[\mathcal{L}^{(l)} > Q] \leq \mathcal{T}_{G}$ and $\text{Pr}[\mathcal{L}^{(l+1)} > Q] > \mathcal{T}_{G}$. Since ${\rm Pr}[\mathcal{L}^{(l)} > Q] \leq \mathcal{T}_{G}$ is equivalent to ${\rm Pr}[\mathcal{L}^{(l)} \leq Q] > \mathcal{T}_{G}$, it is evident that the probability that such a case will happen is $F_{\mathcal{L}^{(l)}}(Q)\mathcal{U}\{\mathcal{T}_{G}\}\overline{F}_{\mathcal{L}^{(l+1)}}(Q)\mathcal{U}\{\mathcal{T}_{G}\}$.
	\item No secondary transmission is allowed {\it iff} ${\rm Pr}[\mathcal{L}^{(1)} > Q] > \mathcal{T}_{G}$. As a result, the probability that such a case will happen is $\overline{F}_{\mathcal{L}^{(1)}}(Q)\mathcal{U}\{\mathcal{T}_{G}\}$.
\end{itemize}

In summary, the probability mass function (PMF) of the active number of secondary Tx antennas is explicitly given by
\begin{align}
\nonumber
&\text{Pr}[M=l]=\\
&\left \{
\begin{tabular}{l l}
$\overline{F}_{\mathcal{L}^{(1)}}(Q)\mathcal{U}\{\mathcal{T}_{G}\},$ & $l=0$\\
$F_{\mathcal{L}^{(l)}}(Q)\mathcal{U}\{\mathcal{T}_{G}\}\overline{F}_{\mathcal{L}^{(l+1)}}(Q)\mathcal{U}\{\mathcal{T}_{G}\},$ & $1\leq l\leq M-1$\\
$F_{\mathcal{L}^{(M)}}(Q)\mathcal{U}\{\mathcal{T}_{G}\},$ & $l=M$
\end{tabular}
\right\}
\label{pmf}
\end{align}

\begin{proposition}
The average number of active secondary Tx antennas can be calculated by
\begin{align}
\nonumber
\overline{M}&=\left(\sum^{M-1}_{l=1}\left[l\:F_{\mathcal{L}^{(l)}}(Q)\mathcal{U}\{\mathcal{T}_{G}\}\overline{F}_{\mathcal{L}^{(l+1)}}(Q)\mathcal{U}\{\mathcal{T}_{G}\}\right]\right)\\
&\ \ \ +M F_{\mathcal{L}^{(M)}}(Q)\mathcal{U}\{\mathcal{T}_{G}\}.
\label{avM}
\end{align}
\end{proposition}

\begin{proof}
By definition, it holds that 
\begin{align}
\overline{M}=\sum^{M}_{l=0}l \ \text{Pr}[M=l].
\label{avM1}
\end{align}
Thus, inserting (\ref{pmf}) into (\ref{avM1}) yields the desired result.
\end{proof}

Since $1-F_{\mathcal{L}^{(l)}}(Q) = \overline{F}_{\mathcal{L}^{(l)}}(Q)$ and $\overline{F}_{\mathcal{L}^{(l)}}(Q)$ can be computed according to \eqref{probinstant1} by setting $M = l$, the average number of active secondary Tx antennas given by \eqref{avM} can be analytically obtained. 

\section{Numerical Results and Discussions}
\label{Numerical Results}
In this section, numerical results are presented and compared with Monte-Carlo (MC) simulation results. In what follows, curves and circle-marks correspond to the analytical and simulation results, respectively. In the ensuing simulation experiments, all the involved average channel gains including $\mathbb{E}[X]$, $\mathbb{E}[Z_{k}]$, $\forall k \in [1, L_{T}]$, and $\mathbb{E}[Y^{(l)}]$, $\forall l \in [1, L_{R}]$, are determined by $(d/d_{\text{ref}})^{-\alpha}$, where $d$ is the Euclidean distance in the unit of meters between two nodes of interest, $d_{\text{ref}}$ is a reference distance of 100 meters (used for normalization), and $\alpha = 4$ denotes the path-loss exponent. In particular, the distance $d$ has three instances: $d_{\rm S_{T}-S_{R}}$, $d_{\rm P_{T}-S_{R}}$ and $d_{\rm S_{T}-P_{R}}$, which denote the distances from ST to SR, from PT to SR and from ST to PR, respectively. The noise variance is normalized, say, $N_{0}=1$, and other parameter setting includes $p_{\max} = 20$ dB, $Q = 7$ dB, $\gamma_{\rm th} = 3$ dB, and $p_{\rm p} = 10$ dB with respect to the normalized noise power.
 
Figure \ref{fig2} depicts the outage probability experienced at SR, where the distances $d_{\rm S_{T}-S_{R}}$ and $d_{\rm P_{T}-S_{R}}$ are fixed to $18$m and $56$m, respectively, while the value of $d_{\rm S_{T}-P_{R}}$ varies from $30$m to $100$m (other system parameter setting is specified in the title of the figure). Note that since channel gains are given in the form of $(d/d_{\rm ref})^{-\alpha}$, $d_{\rm S_{T}-S_{R}} = 25$m corresponds to $\mathbb{E}[X] = 256$, while $d_{\rm P_{T}-S_{R}} = 56$m corresponds to $\mathbb{E}[Z_{\rm i.i.d.}] = 10$. The same methodology is applied in the rest of this section to calculate the distance-dependent parameters. It is observed that, for a certain number of Rx antenna, i.e., $N$, increasing the value of $d_{\rm S_{T}-P_{R}}$, i.e., the distance between ST and PR (or equivalently increasing the IT dictated by primary users), the outage probability decreases significantly, since larger Tx power can be allocated at ST. On the other hand, for a fixed value of $d_{\rm S_{T}-P_{R}}$, increasing the number of Rx antennas (i.e., larger $N$) will decrease the outage probability as well, due to larger spatial reception diversity gain.

\begin{figure}[!t]
\centering
\includegraphics[trim=1.8cm 0.2cm 2.5cm 0cm, clip=true,totalheight=0.28\textheight]{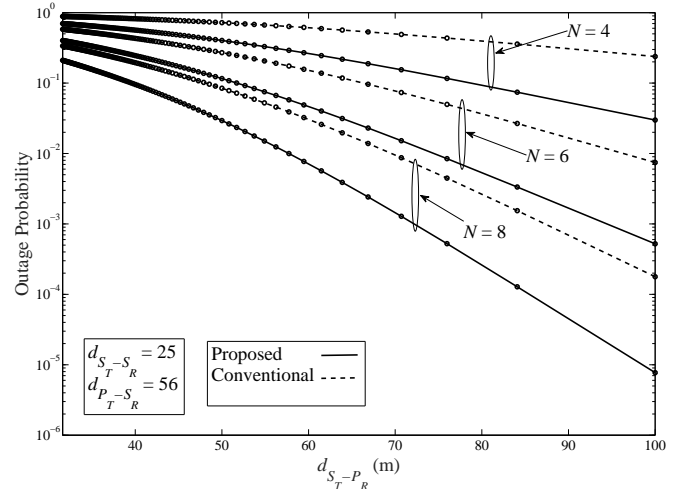}
\caption{Outage probability vs. the distance between ST and PR, where all the secondary-to-primary links are assumed i.i.d. Also, $L_{T} = L_{R} = 2$ and $M = 4$.}
\label{fig2}
\end{figure}

In addition, the benefit of the optimal Tx power allocation is illustrated. In particular, Figs.~\ref{fig2}-\ref{fig4} make a comparison of our optimal power allocation scheme with the conventional one. More specifically, the conventional (fixed) power-allocation can be expressed as
\begin{equation}
p^{(\text{conv.})}_{i}=\min\left\{\frac{Q}{M\mathbb{E}[Y]},\frac{p_{\max}}{M}\right\}.
\label{uncpdfsnrapprox}
\end{equation}
As seen from both Fig.~\ref{fig2} and Fig.~\ref{fig3}, the proposed scheme with optimal power allocation outperforms the conventional scheme. In particular, the superiority becomes more emphatic with larger value of $d_{\rm S_{T}-S_{R}}$, since the effect of power allocation dominates the outage probability when the distance between ST and SR is large. Also, Fig.~\ref{fig3} shows that the outage probability decreases with larger number of Rx antennas, as expected. 

\begin{figure}[!t]
\centering
\includegraphics[trim=1.8cm 0.0cm 2.5cm 0cm, clip=true,totalheight=0.28\textheight]{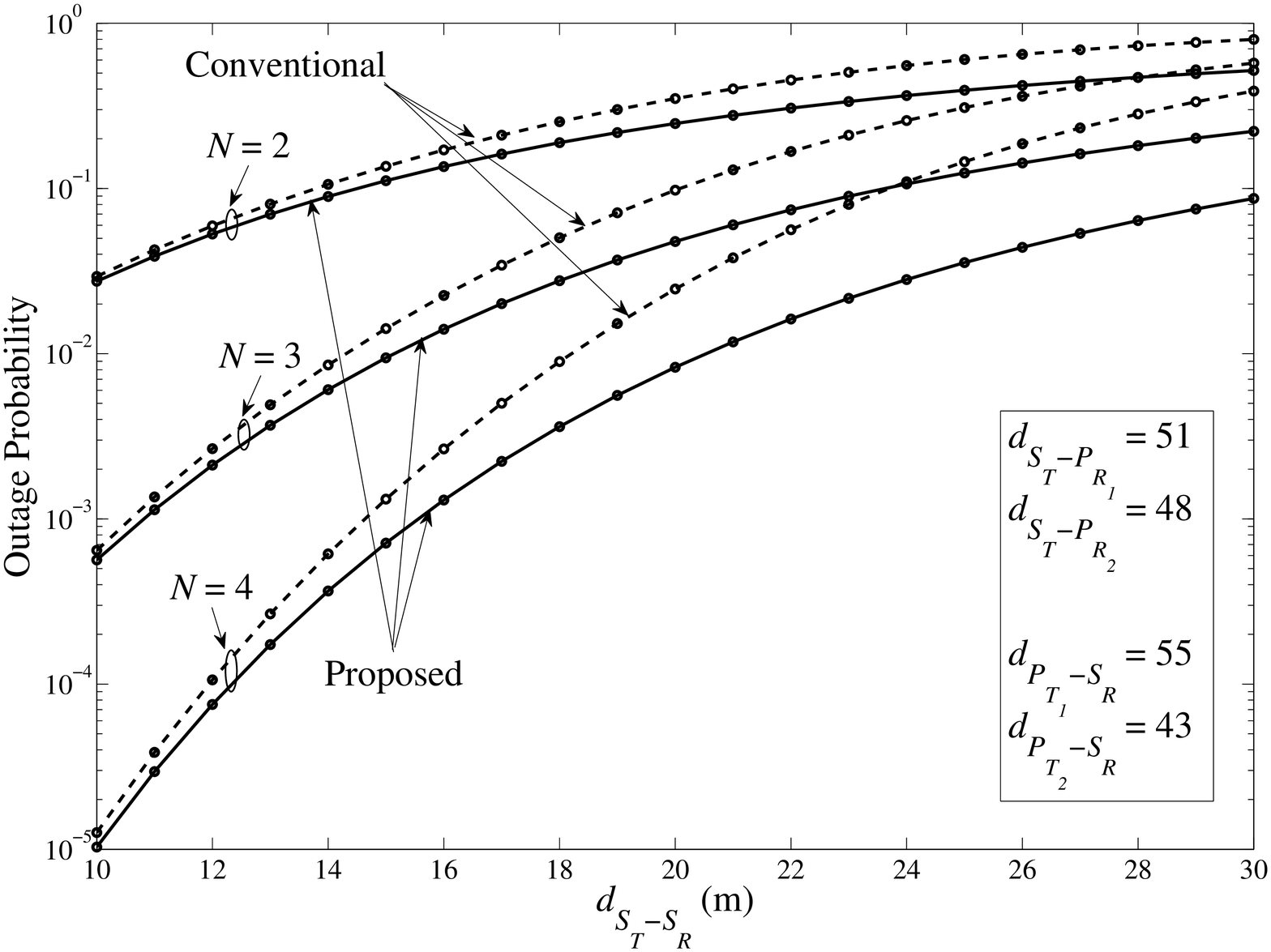}
\caption{Outage probability vs. the distance between ST and SR, where all the secondary-to-primary links are assumed i.n.i.d. Also, $L_{T}=L_{R}=2$ and $M=2$.}
\label{fig3}
\end{figure}

Fig.~\ref{fig4} depicts the outage probability of the secondary transmission in the presence of a large number of Rx antenna (i.e., in the sense of massive MIMO). In the simulation setting, the distance between ST and SR is set to $20$m, which is a typical range of a femtocell deployment \cite{j:femtosurvey}. Similar to the observation from Fig.~\ref{fig3}, the proposed scheme with optimal power allocation outperforms the conventional scheme. In other words, the power-allocation parameter $\lambda$ given by \eqref{langrangecl}, plays a key role to the system performance. In addition, it can be seen from Fig.~\ref{fig4} that the outage probability of the secondary transmission becomes smaller with increasing $N$, which is in agreement with \eqref{sinridistrrr1}. 

\begin{figure}[!t]
\centering
\includegraphics[trim=1.8cm 0.0cm 2.5cm 0cm, clip=true,totalheight=0.28\textheight]{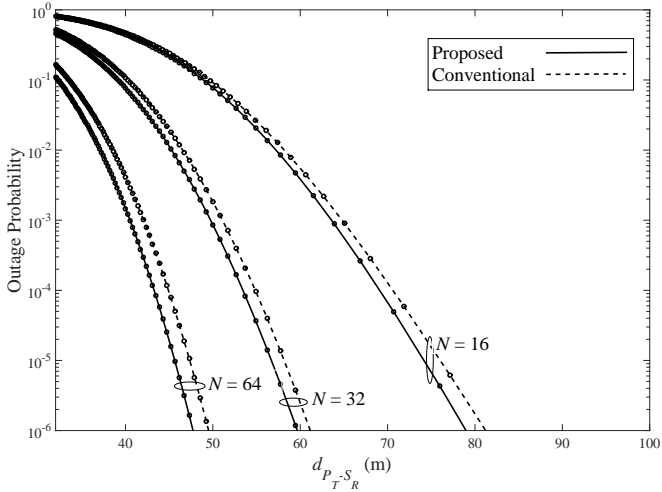}
\caption{Outage probability vs. the distance between PTs and SR, where all the secondary-to-primary links are assumed i.i.d. Also, $L_{T}=L_{R}=2$, $d_{S_{T}-P_{R}}=56$m $d_{S_{T}-S_{R}}=20$m and $M=4$.}
\label{fig4}
\end{figure}

To verify Corollary \ref{cordeterm}, Fig.~\ref{fig6} compares the simulated data rate of each secondary data stream with an equivalent one based on the deterministic SINR given by \eqref{sinridistr1}. Specifically, the average data rate of the $i^{\rm th}$ stream (in bps/Hz) is computed as $\mathcal{R}_{i} = \mathbb{E}[\log_{2}(1+{\rm SINR}_{i})]$ via MC simulations by using \eqref{sinri}. Alternatively, the following semi-analytical approach can be used as
\begin{equation}
\mathcal{R}_{i} 
= \frac{1}{\text{ln}(2)}\int^{\infty}_{0}\frac{1-P^{(i)}_{\text{out}}(x)}{1+x}{\rm d}x,
\end{equation} 
which is more efficient than exhaustive MC simulations. In the case when $N$ and $L_{T}$ approaches infinity, the deterministic SINR of (\ref{sinridistr1}) can be used, such that
\begin{equation}
\mathcal{R}^{(\text{determ.})}_{i}
\triangleq \log_{2}\left(1+\frac{\left(\frac{\lambda \mathbb{E}[X]}{\text{ln}(2)\mathbb{E}[Y]}\right)(N-M+1)}{\frac{p_{\rm p}}{L_{T}}\sum^{L_{T}}_{k=1}\mathbb{E}[Z_{k}]}\right),
\end{equation} 
and
\begin{equation}
\mathcal{R}_{i}-\mathcal{R}^{(\text{determ.})}_{i} \rightarrow 0, \text{ as } N \text{ and } L_{T} \rightarrow +\infty, \ \forall i \in [1, M].
\end{equation}
It is seen from Fig.~\ref{fig6} that the deterministic data rate coincides the actual one (obtained by using the aforementioned semi-analytical approach) when $N = L_{T} \geq 80$ ($M=16$), which corroborates the effectiveness of Corollary 4.

\begin{figure}[!t]
\centering
\includegraphics[trim=1.8cm 0.2cm 2.5cm 0cm, clip=true,totalheight=0.28\textheight]{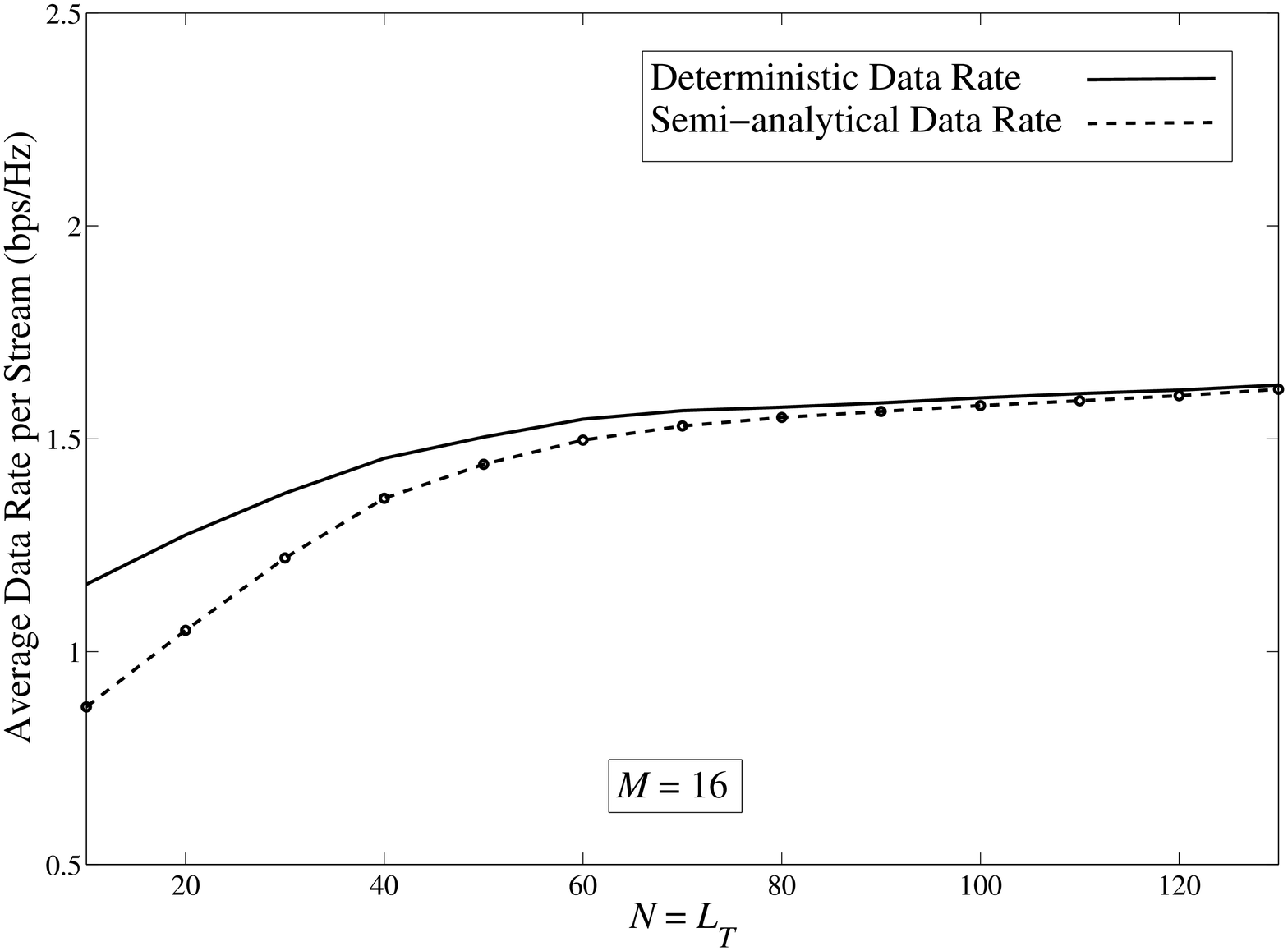}
\caption{Average data rate for each secondary data stream vs. various $\{N=L_{T}\}$ values. The distances between ST and SR are fixed to $15$m, while the distances among the primary and secondary nodes are fixed to $30$m.}
\label{fig6}
\end{figure}

Figure \ref{fig5} shows the average number of active antennas for a secondary link with $M=N=4$ where the proposed antenna reduction algorithm is applied. It is observed that, for a tighter interference leakage constraint (i.e., smaller $\mathcal{T}_{G}$), the average secondary Tx antennas (i.e., $\overline{M}$) decreases in order not to introduce excessive harmful interference to primary users. This is in contrast to the conventional underlay cognitive relaying transmission where it always holds that $\overline{M}=M$, yielding excessive instantaneous interference to the primary users, even though the average interference constraint is satisfied.

\begin{figure}[!t]
\centering
\includegraphics[trim=1.8cm 0.2cm 2.5cm 0cm, clip=true,totalheight=0.28\textheight]{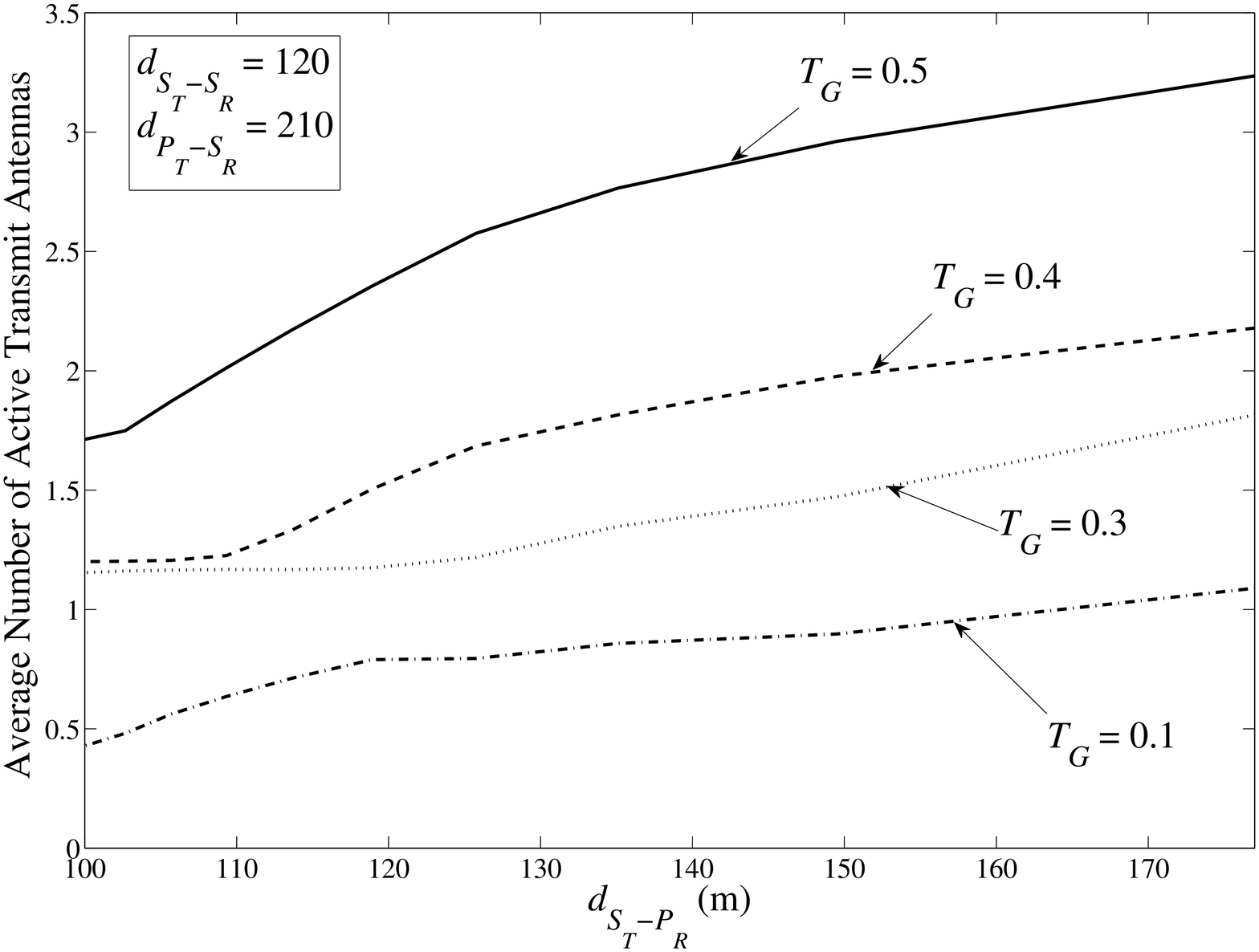}
\caption{Average number of active transmit (secondary) antennas vs. various distances between ST and PR, where all the secondary-to-primary links are assumed i.i.d. Also, $L_{T}=L_{R}=1$ and $M=4$.}
\label{fig5}
\end{figure}

In Fig. \ref{newfig}, the normalized average number of active secondary antennas (i.e., $\overline{M}/M \in [0,1]$) is presented for different system configuration scenarios and for an increasing number of transmit antennas. Obviously, $\overline{M}\rightarrow M$ for higher $M$ values (i.e., when approaching massive MIMO conditions). This occurs due to the channel hardening effect. In other words, the aforementioned ceiling of $p_{i}$ with respect to $\mathbb{E}[Y]$, i.e., (\ref{optproblem1}), gets more tight as $M$ grows, since $Y\rightarrow \mathbb{E}[Y]$ as $M\rightarrow +\infty$. In this case, the effect of unexpected interference leakage to PRs tends to zero (i.e., ${\rm Pr}\left[\mathcal{L}^{(M)} > Q\right]\rightarrow 0^{+}$ as $M\rightarrow +\infty$). However, in practice, $M$ is bounded and the latter effect can play a critical role to the transmission quality of primary service.\footnote{As the right-most part of Fig. \ref{newfig} reveals, the effective number of secondary transmit antennas $M_{E}$ arises from fewer iterations of Algorithm 1 as $M$ increases (on average). Hence, the computational complexity of the proposed scheme is drastically reduced for massive MIMO deployments.} It is also clear from Fig. \ref{newfig} that the aforementioned interference leakage gets more intense for closer primary-to-secondary distances, and vice versa.

\begin{figure}[!t]
\centering
\includegraphics[trim=1.8cm 0.2cm 2.5cm 0cm, clip=true,totalheight=0.28\textheight]{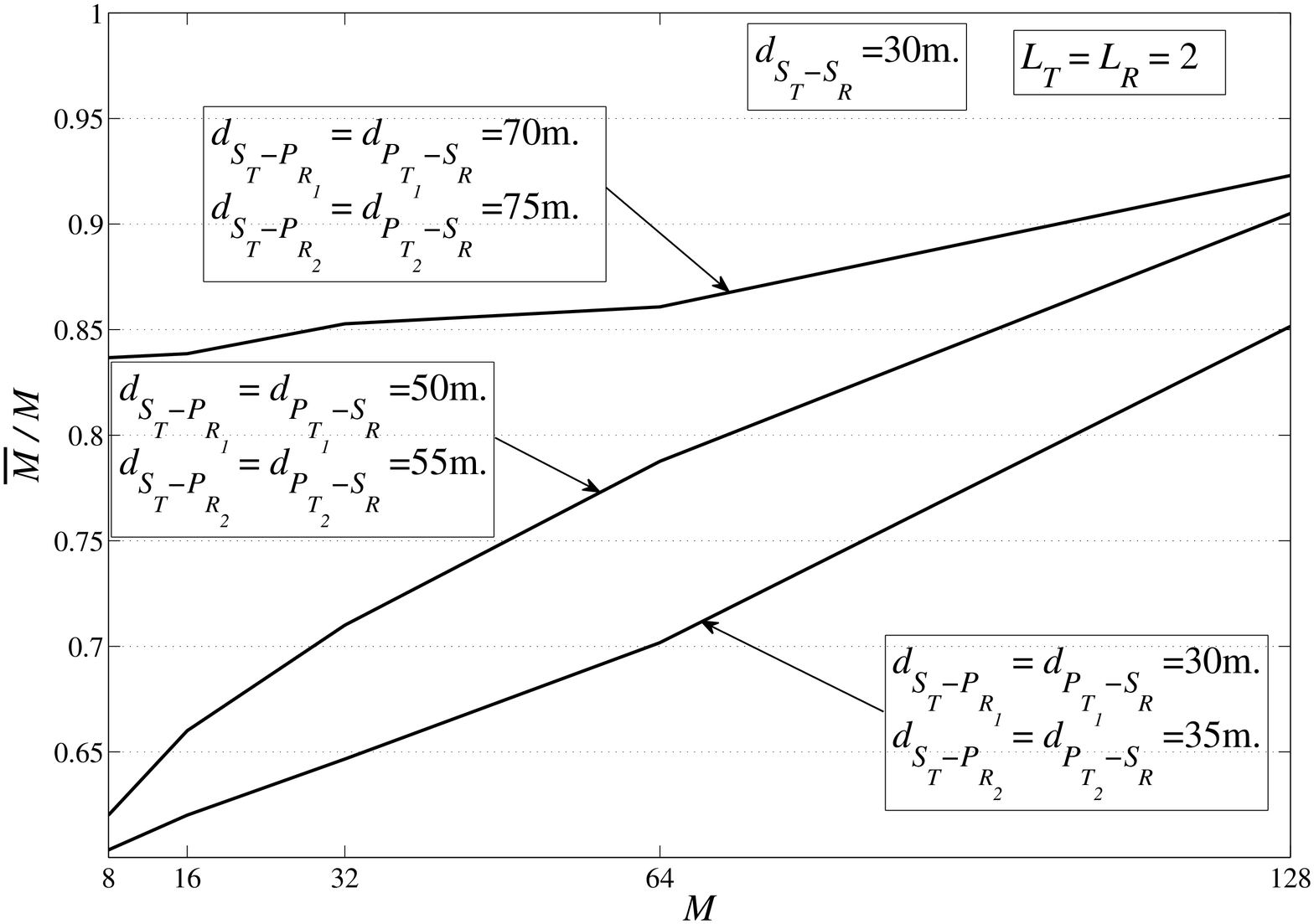}
\caption{Normalized average number of active transmit (secondary) antennas vs. all the available ones, for i.n.i.d. secondary-to-primary links. Also, $\mathcal{T}_{G}=0.1$ is assumed.}
\label{newfig}
\end{figure}

{\color{black}
Finally, Fig.~\ref{fig8} illustrates the beneficial role of the proposed antenna reduction scheme. In particular, the outage performance of the proposed approach (using optimal power allocation, yet with a fixed $M$) is compared with the antenna reduction scheme, which utilizes a versatile $M_{E}$ according to Algorithm~1. A massive MIMO regime is considered for the secondary antenna array, where both $N$ and $M$ have quite high values. For the antenna reduction scheme, the probability of interference leakage regarding the primary system is set to $\mathcal{T}_{G}=0.1$. Obviously, the antenna reduction scheme outperforms the standard scheme with fixed $M$, as expected. This occurs because the effective number of secondary Tx antennas $M_{E}\leq M$ and, thus, the corresponding channel gain (i.e., see \eqref{sinridistr1} and \eqref{sinridistrrrrr1}) becomes  $N-M_{E}+1\geq N-M+1$. Doing so, the resultant $M_{E}$ secondary streams experience higher SINR conditions at the secondary Rx, while mitigating the effect of excessive harmful interference to the primary system at the same time.
}

\begin{figure}[!t]
\centering
\includegraphics[trim=1.8cm 0.2cm 2.5cm 0cm, clip=true,totalheight=0.28\textheight]{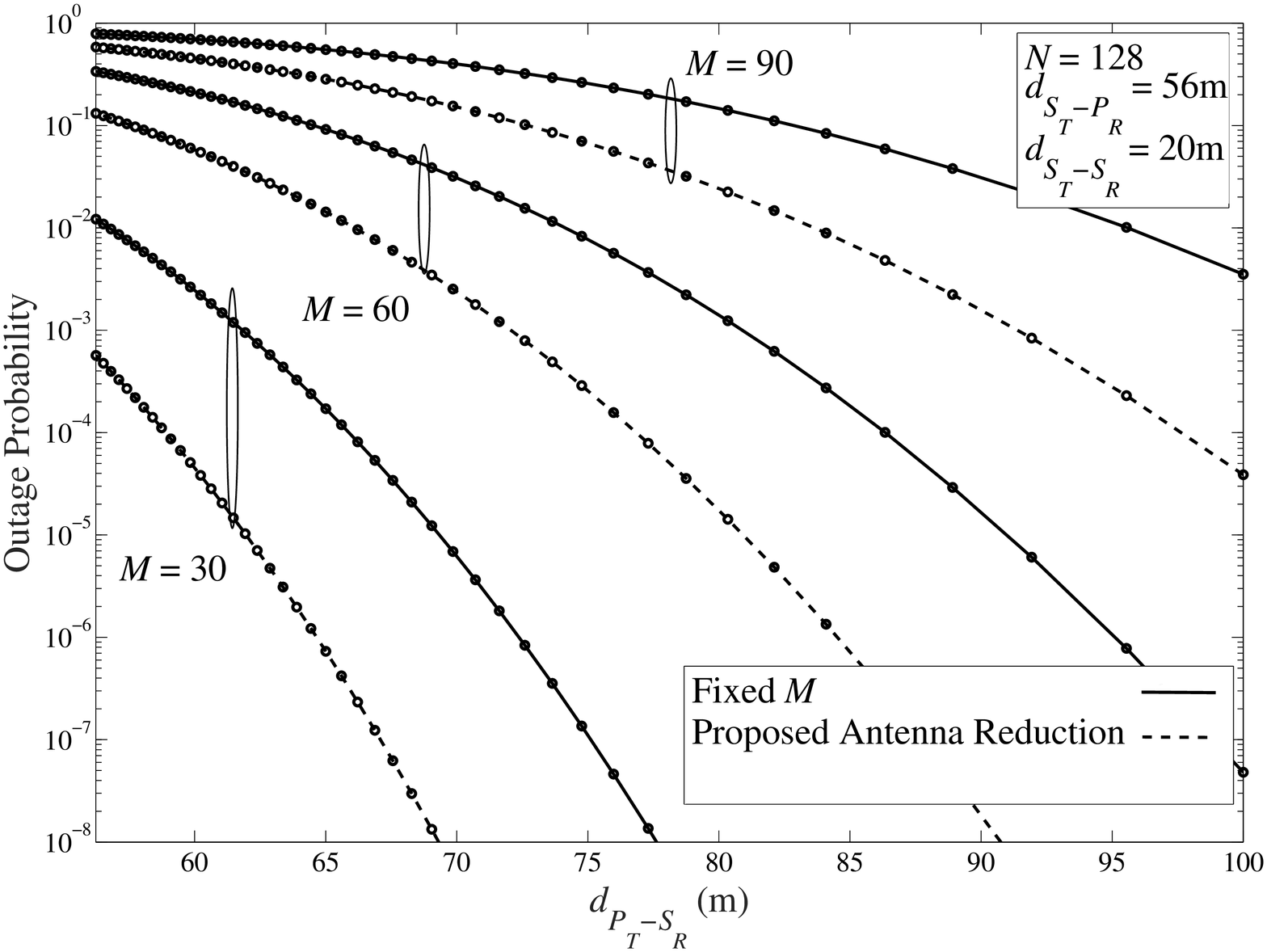}
\caption{Outage probability of the secondary system in the massive antenna array regime vs. various distances between PT and SR, where all the secondary-to-primary links are assumed i.i.d. Also, $\mathcal{T}_{G}=0.1$ is assumed.}
\label{fig8}
\end{figure}

\section{Conclusion}
\label{Conclusion}
The performance of underlay MIMO CR systems was studied, where independent secondary data streams are simultaneously transmitted and received via ZF detection. The analysis included the rather practical scenario of inter-system interference between primary and secondary systems plus AWGN, under independent Rayleigh fading channels. Also, the scenarios of multiple randomly distributed single-antenna and co-located multiple-antenna primary nodes were both considered. An optimal power allocation of the secondary transmission was presented aiming to enhance the received data rate, when only second-order CSI regarding the primary-to-secondary channels is available. Based on this scheme, a new closed-form and exact expression for the outage performance of secondary system was derived. Some special cases of interest were also analyzed, such as the massive MIMO deployments for the secondary and/or primary system. In addition, a new linear and computationally-efficient algorithm was analytically presented, which is able to control the total secondary transmission power so as to better preserve the communication quality of the primary service. The enclosed numerical results verified the accuracy of the analysis as well as the efficacy of the proposed scheme. 

\appendix

\subsection{Derivation of Eq.~\eqref{sinridistr}}
\label{appSINRdistr}
\numberwithin{equation}{subsection}
\setcounter{equation}{0}
From \eqref{sinri}, we have
\begin{equation}
{\rm SINR}_{i} 
= \frac{\left\|\left[\mathbf{G}^{\dagger}\right]_{i}\right\|^{-2}}{p_{\rm p}\frac{\left\|\left[\mathbf{G}^{\dagger}\right]_{i}\mathbf{H}_{\text{p}}\right\|^{2}}{\left\|\left[\mathbf{G}^{\dagger}\right]_{i}\right\|^{2}} + N_{0}}.
\label{sinrii}
\end{equation}
Then, by using \cite[Theorem 1]{j:GOREHeathPaulraj} and \cite[Eq. (10)]{j:MatthaiouZFMassiveMIMO2013}, we obtain
\begin{eqnarray}
\left\|\left[\mathbf{G}^{\dagger}\right]_{i}\right\|^{-2}
& = & \frac{1}{\left[\left(\mathbf{G}^{\mathcal{H}}_{i}\mathbf{G}_{i}\right)^{-1}\right]_{ii}}  \nonumber \\
& = & \frac{\det{\mathbf{G}^{\mathcal{H}}\mathbf{G}}}{\det{\overline{\mathbf{G}}_{i}^{\mathcal{H}}\overline{\mathbf{G}}_{i}}}  \nonumber \\
& = & \mathbf{g}^{\mathcal{H}}_{i}\left(\mathbf{I}_{N}-\overline{\mathbf{G}}_{i}(\overline{\mathbf{G}}^{\mathcal{H}}_{i}\overline{\mathbf{G}}_{i})^{-1}\overline{\mathbf{G}}^{\mathcal{H}}_{i}\right)\mathbf{g}_{i},
\label{zfg}
\end{eqnarray}
where $\overline{\mathbf{G}}_{i}$ stands for the deflated version of $\mathbf{G}$ by removing its $i^{\rm th}$ column.

Let $\mathbf{Q}_{i} \triangleq \mathbf{I}_{N}-\overline{\mathbf{G}}_{i}\left(\overline{\mathbf{G}}^{\mathcal{H}}_{i}\overline{\mathbf{G}}_{i}\right)^{-1}\overline{\mathbf{G}}^{\mathcal{H}}_{i}$. Clearly, $\mathbf{Q}_{i}$ is a $N \times N$ square matrix and represents the projection onto the null space of $\overline{\mathbf{G}}^{\mathcal{H}}_{i}$. In addition, $\mathbf{Q}_{i}$ is Hermitian and idempotent.\footnote{Note that $\overline{\mathbf{G}}_{i}\left(\overline{\mathbf{G}}^{\mathcal{H}}_{i}\overline{\mathbf{G}}_{i}\right)^{-1}\overline{\mathbf{G}}^{\mathcal{H}}_{i}$ is a $N \times N$ symmetric and idempotent (projection) matrix. Thus, $\mathbf{Q}_{i}$ is also an idempotent matrix \cite{b:freedman2009statistical}.} As a result, its eigenvalues are either zero or one and it has a rank of $N-M+1$. In other words, we know that
\begin{equation}
\text{Eigenvalues of }\mathbf{Q}_{i}: \underbrace{0, 0, \cdots, 0}_{M-1}, \underbrace{1, 1, \cdots, 1}_{N-M+1}.
\label{eigen}
\end{equation}
Next, by performing eigenvalue decomposition over $\mathbf{Q}_{i}$, the expression \eqref{zfg} can be rewritten as
\begin{equation}
\left\|\left[\mathbf{G}^{\dagger}\right]_{i}\right\|^{-2} 
= \mathbf{g}^{\mathcal{H}}_{i}\mathbf{Q}_{i}\mathbf{g}_{i} 
= \mathbf{g}^{\mathcal{H}}_{i}\mathbf{U}_{i}\mathbf{\Lambda}_{i}\mathbf{U}^{\mathcal{H}}_{i}\mathbf{g}_{i},
\label{y2}
\end{equation}
where $\mathbf{U}_{i}$ is a unitary matrix and $\mathbf{\Lambda}_{i} = \diag\{\lambda_{1},\cdots,\lambda_{N}\}$ corresponds to the eigenvalues of $\mathbf{Q}_{i}$. Finally, by virtue of \eqref{eigen} and recalling the isotropic property of zero-mean Gaussian vectors \cite[Chapter 1]{b:multivariate}, i.e., $\mathbf{U}^{\mathcal{H}}_{i}\mathbf{g}_{i}\overset{\text{d}}=\mathbf{g}_{i}$, (\ref{y2}) becomes
\begin{equation}
\left\|\left[\mathbf{G}^{\dagger}\right]_{i}\right\|^{-2}
= \sum^{N}_{i=1}\lambda_{i}(\mathbf{U}^{\mathcal{H}}_{i}\mathbf{g}_{i})^{\mathcal{H}}(\mathbf{U}^{\mathcal{H}}_{i}\mathbf{g}_{i})\overset{\text{d}}
= \sum^{N}_{i=1}\lambda_{i}\mathbf{g}^{\mathcal{H}}_{i}\mathbf{g}_{i}
\triangleq p_{i}X_{i}.
\label{y3}
\end{equation}
Notice that $p_{i} X_{i}$ in \eqref{y3} does not reflect the actual value of $\|[\mathbf{G}^{\dagger}]_{i}\|^{-2}$, yet denotes equality in distribution, which is sufficient for subsequent performance analysis.

On the other hand, given $\mathbf{G}^{\dagger}$, $[\mathbf{G}^{\dagger}]_{i}\mathbf{H}_{\text{p}}/\left\|\left[\mathbf{G}^{\dagger}\right]_{i}\right\|$ is a Gaussian vector of $L_{T}$ i.n.i.d. RVs, which is independent of $[\mathbf{G}^{\dagger}]_{i}$ \cite[Theorem 1.5.5]{b:multivariate}. Hence, $Z \triangleq \left\|\left[\mathbf{G}^{\dagger}\right]_{i}\mathbf{H}_{\text{p}}\right\|^{2}/\left\|\left[\mathbf{G}^{\dagger}\right]_{i}\right\|^{2}$ in the denominator of \eqref{sinrii} is equivalently composed of $L_{T}$ i.n.i.d. exponential RVs, whose statistical means reflect the path losses between PTs and SR.

\subsection{Derivations of Eqs.~\eqref{EY} and \eqref{EZ}}
\label{appEYderiv}
\numberwithin{equation}{subsection}
\setcounter{equation}{0}
The interfering channel gain from the $i^{\rm th}$ antenna of ST to the $l^{\rm th}$ single-antenna PR is distributed as $|y^{(l)}_{i}|^{2}\overset{\text{d}} = \exp\left(-\mathbb{E}[Y^{(l)}_{i}]\right)$, with $\mathbb{E}[Y^{(l)}_{i}]$ being its mean. As aforementioned, $|y^{(l)}_{i}|^{2}$, $\forall i\in [1, M]$, are identically distributed, such that $\mathbb{E}[Y^{(l)}_{i}]\triangleq \mathbb{E}[Y^{(l)}]$. Let $Y \triangleq \max_{j}\{y^{(j)}\}^{L_{R}}_{j=1}$, we have that
\begin{equation}
\mathbb{E}[Y]
= \int^{\infty}_{0} {y f_{Y}(y)}{\rm d}y,
\label{EEEY}
\end{equation}
where it follows from \cite[Eq. (A.5)]{DBLP:journals/corr/MiridakisTAD15} that
\begin{eqnarray}
\nonumber
f_{Y}(y)
& \hspace{-5pt} = \hspace{-5pt} &\sum^{L_{R}}_{l=1}\sum^{L_{R}}_{k=0}\underbrace{\sum^{L_{R}}_{n_{1}=1}\cdots \sum^{L_{R}}_{n_{k}=1}}_{n_{1}\neq \cdots \neq n_{k}\neq k}\frac{(-1)^{k}}{k!\mathbb{E}[Y^{(l)}]}\\
&    &{}\times \exp\left(-\left(\frac{1}{\mathbb{E}[Y^{(l)}]}+\sum^{k}_{t=1}\frac{1}{\mathbb{E}[Y^{(n_{t})}]}\right)y\right).
\label{fY}
\end{eqnarray}
Notice that $Y$ stands for the maximal value of $L_{R}$ independent yet non-identical exponentially distributed RVs (due to arbitrarily different distances between ST and $L_{R}$ PRs). Substituting \eqref{fY} into \eqref{EEEY} and performing some algebraic manipulations, the desired result in \eqref{EY} can be derived.

Moreover, it holds from \cite[Eq. (5)]{j:rayleighinid} that
\begin{equation}
f_{Z}(z) 
= \sum^{L_{T}}_{k=1}\left(\prod^{k}_{j=1,j\neq k}\frac{\mathbb{E}[Z_{k}]}{\left(\mathbb{E}[Z_{k}]-\mathbb{E}[Z_{j}]\right)}\right)\frac{\exp\left(-\frac{z}{\mathbb{E}[Z_{k}]}\right)}{\mathbb{E}[Z_{k}]},
\label{fZ}
\end{equation}
which yields \eqref{EZ} by using that $\mathbb{E}[Z] = \int^{\infty}_{0}z f_{Z}(z) {\rm d}z$.

\subsection{Derivation of Eq.~\eqref{langrangecl}}
\label{applagrange}
\numberwithin{equation}{subsection}
\setcounter{equation}{0}
From \eqref{langrange}, we know that
\begin{align}
\nonumber
&\mathbb{E}_{X_{i}}[p_{i}] = \min\left\{\frac{Q}{M\mathbb{E}[Y]},\frac{p_{\max}}{M}\right\}
\Longleftrightarrow \\ 
\nonumber
&\int^{\infty}_{\mathcal{C}}\left(\frac{\lambda}{\text{ln}(2)\mathbb{E}[Y]}-\frac{\left(p_{\rm p}\mathbb{E}[Z]+N_{0}\right)}{x}\right)f_{X_{i}}(x){\rm d}x\\
& = \min\left\{\frac{Q}{M\mathbb{E}[Y]},\frac{p_{\max}}{M}\right\}.
\label{langrangee}
\end{align}
Note that the parameter $\min\{Q/(M\mathbb{E}[Y]),p_{\max}/M\}$ is used for clipping. Particularly, it ensures that in the case of very far-distant primary nodes (i.e., when $\mathbb{E}[Y]\rightarrow 0^{+}$), $p_{\max}/M$ is used for an appropriate upper bound on the secondary transmission at each antenna, as in conventional (non-cognitive) MIMO systems. In addition, the aforementioned integration limit follows the restriction
\begin{equation}
0 \leq p^{\star}_{i} \Longleftrightarrow \mathcal{C} \leq X_{i}, \forall i \in [1, M].
\end{equation}
Therefore, noticing from (\ref{xidistr}) that $X_{i}$ follows an Erlang distribution with shape parameter $N-M+1$ and scale parameter $\mathbb{E}[X]$, inserting \cite[Eq. (3.381.3)]{tables} into \eqref{langrangee} and performing some algebraic manipulations, we attain \eqref{langrangecl}.

\subsection{Derivation of Eq.~\eqref{uncpdfsnr}}
\label{apppdfsnr}
\numberwithin{equation}{subsection}
\setcounter{equation}{0}
The CDF of $p^{\star}_{i}X_{i}$ is explicitly defined as
\begin{align}
\nonumber
F_{p^{\star}_{i}X_{i}}(\gamma)&\triangleq \text{Pr}\left[\left(p^{\star}_{i}X_{i}\right)\leq \gamma\right]\\
&=\text{Pr}\left[X_{i}\leq \frac{\text{ln}(2)\mathbb{E}[Y]\gamma}{\lambda}+\mathcal{C}\right],\ \ \gamma>0.
\label{Fsnri}
\end{align}
According to \eqref{xidistr}, the corresponding PDF and CDF of $X_{i}$ are, respectively, given by
\begin{align}
f_{X_{i}}(x)=\frac{x^{N-M}\exp\left(-\frac{x}{\mathbb{E}[X]}\right)}{(N-M)!\mathbb{E}[X]^{N-M+1}},
\label{fxi}
\end{align}
and
\begin{align}
F_{X_{i}}(x)=1-\exp\left(-\frac{x}{\mathbb{E}[X]}\right)\sum^{N-M}_{l=0}\frac{\left(\frac{x}{\mathbb{E}[X]}\right)^{l}}{l!}.
\label{fsnri}
\end{align}
With the resulting \eqref{Fsnri} and \eqref{fsnri}, \eqref{uncpdfsnr} can be readily obtained.

\subsection{Derivation of Eq.~\eqref{poutcl}}
\label{apppoutcl}
\numberwithin{equation}{subsection}
\setcounter{equation}{0}
The CDF of the received SINR pertaining to the $i^{\rm th}$ secondary data stream can be given by
\begin{eqnarray}
F_{{\rm SINR}_{i}}(\gamma)
& = &\text{Pr}\left[\frac{p^{\star}_{i}X_{i}}{p_{\rm p}Z+N_{0}}<\gamma\right] \nonumber \\
& = &\text{Pr}\left[p^{\star}_{i}X_{i}<\gamma (p_{\rm p}Z+N_{0})\right] \nonumber\\
& = &\int^{\infty}_{0}F_{p^{\star}_{i}X_{i}}\left(\gamma (p_{\rm p}z + N_{0})\right)f_{Z}(z){\rm d}z.
\label{cdfsinr}
\end{eqnarray}
Substituting \eqref{uncpdfsnr} and \eqref{fZ} into \eqref{cdfsinr} and utilizing \cite[Eq. (3.382.4)]{tables} as well as noticing that $P^{(i)}_{\rm out}(\gamma_{\rm th})\triangleq F_{{\rm SINR}_{i}}(\gamma_{\rm th})$, we arrive at \eqref{poutcl}.

\bibliographystyle{IEEEtran}
\bibliography{IEEEabrv,References}

\vfill

\end{document}